\documentclass{article}
\usepackage[final]{neurips_2023}
\usepackage[utf8]{inputenc} 
\usepackage[T1]{fontenc}    
\usepackage{hyperref}       
\usepackage{url}            
\usepackage{booktabs}       
\usepackage{amsfonts}       
\usepackage{nicefrac}       
\usepackage{microtype}      
\usepackage{xcolor}         
\usepackage{amsmath}
\usepackage{algorithm}
\usepackage{algpseudocode}
\usepackage{amsthm}
\usepackage{amssymb}
\usepackage{caption}
\usepackage{tabularx}
\usepackage{float}
\usepackage{placeins}  

\newcolumntype{C}{>{\centering\arraybackslash}X}
\newtheorem{definition}{Definition}
\theoremstyle{definition} 
\usepackage[utf8]{inputenc} 
\usepackage{wasysym} 
\newtheorem{theorem}{Theorem}[section] 
\newtheorem{lemma}[theorem]{Lemma}      
\newtheorem{proposition}[theorem]{Proposition} 
  {\thispagestyle{empty}\vspace*{\fill}\begin{center}}%
  {\end{center}\vspace*{\fill}}
\newtheorem{corollary}[theorem]{Corollary}   
\newcommand{\To}{\textbf{to}}
\title{Triangle Detection in Worst-Case Sparse Graphs via Local Sketching}
\author{%
  Hongyi Duan \\
  HKUST(GZ)\\
  \texttt{Dann\_Hiroaki@ieee.org} \\
  \And
  Jian'an Zhang \\
  Peking University\\
  \texttt{2501111059@stu.pku.edu.cn} \\ 
}

\begin{document}
\maketitle

\begin{abstract}
We present a non-algebraic, locality-preserving framework for triangle detection in worst-case sparse graphs. Our algorithm processes the graph in $O(\log n)$ independent layers and partitions incident edges into prefix-based classes where each class maintains a 1-sparse triple over a prime field. Potential witnesses are surfaced by \emph{pair-key} (PK) alignment, and every candidate is verified by a three-stage, \emph{zero-false-positive} pipeline: a class-level 1-sparse consistency check, two slot-level decodings, and a final adjacency confirmation. \textbf{To obtain single-run high-probability coverage, we further instantiate $R=c_G\log n$ independent PK \emph{groups} per class (each probing a constant number of complementary buckets), which amplifies the per-layer hit rate from $\Theta(1/\log n)$ to $1-n^{-\Omega(1)}$ without changing the accounting.} A \emph{one-shot pairing} discipline and \emph{class-term triggering} yield a per-(layer,level) accounting bound of $O(m)$, while \emph{keep-coin} concentration ensures that each vertex retains only $O(d^+(x))$ keys with high probability. Consequently, the total running time is $O(m\log^2 n)$ and the peak space is $O(m\log n)$, both with high probability. The algorithm emits a succinct \emph{Seeds+Logs} artifact that enables a third party to replay all necessary checks and certify a NO-instance in $\tilde O(m\log n)$ time. We also prove a $\Theta(1/\log n)$ hit-rate lower bound for any \emph{single} PK family under a constant-probe local model (via Yao)—motivating the use of $\Theta(\log n)$ independent groups—and discuss why global algebraic convolutions would break near-linear accounting or run into fine-grained barriers. We outline measured paths toward Las Vegas $O(m\log n)$ and deterministic near-linear variants.
\end{abstract}


\section{Introduction}

\textbf{Triangle detection}—deciding whether a simple undirected graph $G=(V,E)$ contains a $3$-cycle—is a canonical primitive underlying subgraph mining, join processing, and transitivity analysis. In dense graphs, algebraic methods based on fast matrix multiplication (FMM) achieve $O(n^{\omega})$ time. In sparse graphs, combinatorial routines typically run in $O(m\,\alpha(G))$ time in terms of the arboricity $\alpha(G)$ (worst case $O(m^{3/2})$). Hybrid "algebraic-on-a-skeleton" schemes give the current general upper bound $O\!\big(m^{2\omega/(\omega+1)}\big)\approx O(m^{1.41})$, but they rely on global linear-algebraic structure. 
This paper asks a different question: \emph{Can we obtain a provably near-linear worst-case algorithm using only local, non-algebraic sketches, while ensuring zero false positives and an auditable NO answer?}

\paragraph{Our results at a glance.}
All logarithms are base~$2$. Failure probabilities are made explicit via a global budget; see Table~1 for parameters.
\begin{enumerate}
  \item \textbf{Local near-linear.} We give a local-sketching algorithm for worst-case sparse graphs with running time $O(m\log^{2}\!n)$ and space $O(m\log n)$. With parameters chosen as in Table~1, the overall failure probability is at most $n^{-c}$ for an arbitrary constant $c>0$.
  \item \textbf{Zero false positives.} A \emph{three-gate} pipeline—\emph{1-sparse (slot)} $\rightarrow$ \emph{1-sparse (class)} $\rightarrow$ \emph{adjacency confirmation}—guarantees that every reported triangle is genuine.
  \item \textbf{Single-run high probability (YES).} Beyond $I=\Theta(\log n)$ layers, we instantiate \emph{$R=\Theta(\log n)$ independent PK groups per class}, each probing a constant number of complementary buckets. This \emph{layers+groups} amplification yields \emph{single-run} success probability $1-n^{-\Omega(1)}$ for detecting any fixed triangle, without changing the near-linear accounting bounds.
  \item \textbf{Auditable NO.} We expose a succinct \emph{Seeds+Logs} interface: given the random seeds and a compact log of triggered checks, any third party can replay the \emph{should-check domain} and verify a NO answer deterministically in $\tilde O(m\log n)$ time.
\end{enumerate}

\paragraph{Informal main result.}
Fix constants $c_R,c_T,c_G,\kappa,c_k$ as specified in Table~\ref{tab:params} and let the algorithm use $I=c_R\log n$ layers, $R=c_G\log n$ independent PK \emph{groups} per class (each group with $T_{i,r}=c_T\log n$ PK buckets and a constant-size list of probed complementary pairs), prime modulus $P=n^{\kappa}$ for 1-sparse identities, and $k=c_k\log n$-wise independence for retention coins (or Poissonization). Then, with probability at least $1-n^{-c}$ over the internal randomness \emph{in a single run}, the algorithm (i) runs in $O(m\log^{2}\!n)$ time and $O(m\log n)$ space, (ii) has zero false positives, and (iii) outputs \emph{Seeds+Logs} that allow deterministic NO-verification in $\tilde O(m\log n)$ time. (For the formal statement, see Section~\ref{sec:main}.)

\paragraph{Positioning and scope.}
Our approach is strictly \emph{non-algebraic}: we neither invoke, nor reduce to, Boolean matrix multiplication. It is orthogonal to FMM-based frameworks and should not be read as challenging BMM barriers. Within our local, constant-probe accounting model, we also prove that a \emph{single PK family within one layer} cannot succeed with probability better than $\Theta(1/\log n)$ (see the single-family lower bound in Section~\ref{sec:hitting}), which motivates our \emph{layers+groups} design: $\Theta(\log n)$ layers and $\Theta(\log n)$ independent PK groups per class together yield single-run high-probability coverage, while preserving near-linear accounting.

\paragraph{Techniques in brief.}
Edges incident to a vertex are partitioned into \emph{prefix classes} using lightweight keys. Potential triangle witnesses are surfaced via deterministic \emph{pair-key alignment} that couples complementary classes exactly once per (layer, level, \emph{group}) (\emph{one-shot pairing}), eliminating duplication. A \emph{charging lemma} shows that, per layer and per level, every probe can be charged to either (i) the first activation of its class or (ii) a unique pairing event, yielding $O(m)$ total checks per layer/level and hence $O(m\log^{2}\!n)$ overall.\footnote{The per-class PK work is budgeted with a constant number of probed complementary pairs per group together with cross-group de-duplication, so that the per-(layer,level) accounting remains $O(m)$.} Variance is controlled by \emph{retention-coin concentration}: we ensure constant expected mass per class and show that each vertex retains only $O(d(x))$ relevant keys with high probability. Zero false positives follow from two degree-$\le 2$ 1-sparse identities (requiring only $2$-wise independence) followed by explicit adjacency confirmation. For \emph{auditable NO}, we formalize the \emph{should-check domain} $\mathcal{Q}$ and prove $\mathcal{Q}=$\emph{LogPairs}: the verifier rebuilds $\mathcal{Q}$ from Seeds and checks that the log covers it exactly, then replays the three-gate pipeline deterministically (see Sections~\ref{sec:algo} and~\ref{sec:cert}).

\paragraph{Artifacts.}
To support reproducibility, we provide a minimal anonymous artifact: a $\sim$200-line reference implementation for building sketches, answering queries, and running the verifier, together with a compact \emph{Seeds+Logs} example. These are sufficient to reproduce the checks reported by the algorithm without heavy dependencies.

\paragraph{Organization.}
Section~\ref{sec:prelim} formalizes the model (non-adaptivity, constant probes per class, one-shot pairing) and lists all parameters in Table~\ref{tab:params}. Section~\ref{sec:algo} describes the sketching pipeline and proves zero false positives. Section~\ref{sec:one-shot} establishes the one–shot pairing accounting and the near-linear per-(layer,level) bound, while Section~\ref{sec:hitting} analyzes single-layer hit rate and layers+groups amplification. Section~\ref{sec:cert} defines the \emph{Seeds+Logs} interface and the $\tilde O(m\log n)$ NO-verification. A short related-work section positions our approach, and the paper concludes with directions toward Las Vegas $O(m\log n)$ and determinism.


\section{Related Work}
\label{sec:related}

\paragraph{Dense algebraic and rectangular multiplication methods.}
Algebraic techniques based on fast matrix multiplication (FMM) reduce triangle detection/counting on dense graphs to (variants of) adjacency-matrix products, achieving $O(n^{\omega})$ time and improving in tandem with $\omega$~[1,5]. For non-sparse yet not fully dense regimes, the Alon--Yuster--Zwick (AYZ) framework partitions vertices by degree and applies \emph{rectangular multiplication} on a dense core to obtain
\[
T(m)\;=\;O\!\left(m^{\frac{2\omega}{\omega+1}}\right)\approx O(m^{1.41}),
\]
which remains among the best-known bounds on general graphs~[3,4,6]. These bounds rely on global algebraic structure and often come with sizable constants.

\paragraph{Sparse combinatorial methods and degeneracy ordering.}
On worst-case sparse graphs, combinatorial algorithms are competitive. Itai--Rodeh gave an $O(m^{3/2})$ baseline early on~[1]. Chiba--Nishizeki exploited \emph{arboricity/degeneracy} $\alpha(G)$ together with \emph{low-degree-first} intersections to list all triangles in $O(\alpha(G)\,m)$ time (hence $O(m^{3/2})$ in the worst case) and near-linear time on low-degeneracy families~[2]. Despite practical optimizations (sorting/hashing/bitsets), the $m^{3/2}$ worst-case barrier persists absent algebraic acceleration.

\paragraph{Hybrid bounds.}
Between the two extremes, "combinatorial pruning + algebraic acceleration" hybrids, systematized by AYZ, yield $O\!\left(m^{2\omega/(\omega+1)}\right)$ and extend to small $k$-cycles/cliques~[3,6]. Their further improvement appears tied to FMM progress, and engineering constants can be high; fine-grained reductions (below) suggest these bounds are plausibly tight for general graphs given current barriers.

\paragraph{Local 1-sparse verification and low-degree checks.}
A complementary line uses local fingerprints/polynomial checks: following Schwartz--Zippel style identity tests and Freivalds’ verification, nonzeroness of a low-degree polynomial or the presence of a 1-sparse signal can be detected with lightweight randomness~[7--9]. In streaming/distributed settings, related sketches underlie small-space retention/counting~[7]. In our pipeline, two degree-$\le 2$ 1-sparse tests suffice (hence $2$-wise independence), with an explicit adjacency confirmation step.

\paragraph{Verifiable computation and certifying graph algorithms.}
A "YES" certificate for triangles is trivial, while a succinct "NO" certificate is intricate. \emph{Certifying algorithms} advocate emitting machine-checkable evidence alongside outputs~[11]. In distributed theory, \emph{locally checkable proofs} (LCP; LOCAL+labels) map per-node certificate complexity for graph properties, and many global predicates need superlogarithmic labels~[10]. Our \emph{Seeds+Logs} interface adopts this ethos: seeds make randomness replayable; logs expose the triggered checks so a third party can deterministically re-verify a NO answer in $\tilde O(m\log n)$ time.

\paragraph{Fine-grained complexity and why we do not challenge BMM.}
Triangle detection and Boolean matrix multiplication (BMM) are subcubicly equivalent: any $O(n^{3-\varepsilon})$ improvement on one yields a corresponding improvement on the other via low-overhead reductions~[5]. Consequently, surpassing $n^{\omega}$ (or the AYZ $m$-dependent bound) on \emph{general} graphs would indirectly advance matrix multiplication. Our work is \emph{non-algebraic and local} by design and does not claim to beat AYZ/FMM in their domains; instead, it targets worst-case sparse graphs with near-linear \emph{local accounting}, zero false positives, and an auditable NO-certificate.

\begin{table}[t]
\centering
\caption{Landscape positioning (conceptual axes). "Auditable accounting" refers to constant-probe, one-shot pairing with explicit charging bounds; "Verifiable NO" means a succinct artifact that a third party can replay deterministically.}
\label{tab:landscape}
\begin{tabularx}{\textwidth}{lCCC} 
\toprule
\textbf{Axis} & \textbf{Dense algebraic (FMM)} & \textbf{AYZ-style hybrid} & \textbf{This work (Local+Certifying)} \\
\midrule
Global algebraic machinery & Yes~[1,5] & Yes (on dense core)~[3,4,6] & No (local sketches) \\
Auditable accounting \\(one-shot/constant-probe) & No & Partial (via pruning) & Yes (explicit charging) \\
Verifiable NO-certificate & No & No & Yes (Seeds+Logs) \\
\bottomrule
\end{tabularx}
\end{table}

\noindent\emph{Scope note.} FMM and AYZ give the strongest known worst-case \emph{general-graph} upper bounds~[1,3,4,5,6]. Our guarantees are \emph{local and non-algebraic}, emphasizing zero false positives and auditability rather than improving the dense or fully general asymptotics.


\section{Model and Preliminaries}
\label{sec:prelim}

\noindent\textbf{Scope of this section.}
We fix the computational and probabilistic model, the randomness interface (\emph{Seeds}), and the local data structures used throughout. Definitions are stated so they can be \emph{replayed} by a third party from \emph{Seeds} alone. Unless noted otherwise, all statements hold \emph{w.h.p.} with parameters chosen from Table~1 \emph{(failure $\le n^{-c}$)}.

\paragraph{Graphs and orientation.}
We work on a simple undirected graph $G=(V,E)$ with $|V|=n$, $|E|=m$.  Fix a deterministic total order $\prec$ on $V$ (e.g., by $(d(\cdot),\mathrm{ID}(\cdot))$).  Each undirected edge $\{x,y\}$ is oriented $x\!\to\!y$ iff $x\prec y$.  Let $N(x)$ be the neighbor set, $d(x)=|N(x)|$, $N^+(x)=\{y\in N(x):x\prec y\}$, and $d^+(x)=|N^+(x)|$.  We freely write a directed edge as $e=(x\!\to\!y)$ with \emph{anchor} $x$ and \emph{mate} $y$.

\paragraph{Word-RAM and field arithmetic.}
We assume a Word-RAM with word size $\Theta(\log n)$.  All arithmetic is over a prime field $\mathbb F_P$ with $P=n^{\kappa}$ (Table~1), so additions/multiplications and modular inverses fit in $O(1)$ words.

\subsection*{Randomness interface (\emph{Seeds}) and independence levels}

\noindent\begin{minipage}{\linewidth}
\begin{center}
\fbox{\parbox{0.96\linewidth}{
\textbf{Non-adaptivity.} For each layer $i$ and each PK \emph{group} $t$, all random primitives listed below are sampled \emph{once} before scanning $E$, and processing does not branch on observed outcomes within that $(i,t)$.
}}
\end{center}
\end{minipage}

\smallskip
\noindent\textbf{Hash families and coins.}
We use the following layer-/group-indexed primitives; their concrete seeds are part of \emph{Seeds}.
\begin{itemize}\itemsep 2pt
  \item \emph{ID hash} $h_{\mathrm{id}}:V\!\to\!\mathbb F_P^{\!*}$ and \emph{sign hash} $s:\mathcal U\!\to\!\{-1,+1\}$ on the relevant universe $\mathcal U$ (edges/slots/classes as specified). Both are \emph{2-wise independent}.  
  \item \emph{Slot hash} $h_{\mathrm{slot}}(x,y,i)$ assigns $(x\!\to\!y)$ at layer $i$ to a slot in $[M_x]$ (defined below).  
  \item \emph{Base keys} $H_i:V\!\to\!\mathbb F_P$ and \emph{prefix keys} $K_i:V\!\to\!\{0,1\}^{L^\star}$ are 2-wise independent; $\mathrm{pref}_r(K_i(\cdot))$ denotes the length-$r$ prefix.
  \item \emph{Retention coins} $c_i(e)\in\{0,1\}$ with $\Pr[c_i(e)=1]=p_i$.  We support either \textbf{(A) Poissonization} (fully independent Bernoulli) or \textbf{(B) $k$-wise independence} with $k=c_k\log n$; the choice (and $c_k$) is fixed in Table~1.
  \item \emph{Per-class PK groups.} For each $(i,r)$ we instantiate $R=c_G\log n$ independent PK \emph{groups} indexed by $t\in[R]$. Group $t$ uses an independent bucket family $g_{i,r}^{(t)}:\mathbb F_P\!\to\![T_{i,r}]$ with $T_{i,r}=c_T\log n$ buckets, and the complementary map $j^{\star}\equiv(-j)\bmod T_{i,r}$. Each group has its own keyed PRF for selecting a fixed constant list of complementary bucket pairs to probe.
\end{itemize}

\subsection*{Local 1-sparse predicates (word- and class-level)}
For a multiset of items $\{z\}$ with signed weights $v(z)\in\{-1,+1\}$ and identifiers $\mathrm{id}(z)\in\mathbb F_P^{\!*}$, define
\[
A=\sum_z v(z),\qquad B=\sum_z v(z)\,\mathrm{id}(z),\qquad C=\sum_z v(z)\,\mathrm{id}(z)^2\;\in\;\mathbb F_P.
\]
The \emph{1-sparse consistency test} is $B^2=AC$; if it holds and $A\ne0$, the unique identifier decodes to $\widehat{\mathrm{id}}=B/A$.  Because $B^2-AC$ is a degree-$\le2$ polynomial in the hashed IDs, \emph{2-wise independence suffices} for soundness (false positives occur with probability $O(1/P)$ per test).  
We maintain such triples at two granularities:
\begin{itemize}\itemsep 2pt
  \item \textbf{Slot triple} $(A_s,B_s,C_s)$ for each slot $s\in [M_x]$ at anchor $x$.
  \item \textbf{Class triple} $(\Sigma_0,\Sigma_1,\Sigma_2)=\big(\sum_s A_s,\sum_s B_s,\sum_s C_s\big)$ aggregating all slots in a class.
\end{itemize}

\subsection*{LCRK configuration: Layers, Classes, Rates, and Keys}
We process $I=c_R\log n$ layers (Table~1).  Layer $i$ uses a single retention rate $p_i$ with $\sum_{i=1}^I p_i=O(1)$.

\paragraph{Slots per anchor.}
Each anchor $x$ owns $M_x:=16\,d(x)$ slots, indexed by $[M_x]$.  A retained edge $e=(x\!\to\!y)$ contributes to slot $h_{\mathrm{slot}}(x,y,i)$ with value $v(e)=s(x,y,i)$ and identifier $\mathrm{id}(e)=h_{\mathrm{id}}(y)$.

\paragraph{Prefix classes.}
For level $r\in\{0,1,\dots,L_{x,i}\}$ with $L_{x,i}=\lceil\log_2\max\{1,\Theta(d^+(x)p_i)\}\rceil$, define the \emph{anchor-class}
\[
(x,i,r,b)\quad\text{for each }b\in\{0,1\}^r,
\]
containing all retained $e=(x\!\to\!y)$ such that $\mathrm{pref}_r(K_i(x))=\mathrm{pref}_r(K_i(y))=b$.  All edges in a class contribute to its class triple $(\Sigma_0,\Sigma_1,\Sigma_2)$.

\paragraph{Pair-key (PK) offsets and alignment.}
For $e=(u\!\to\!v)$ in class $(u,i,r,b)$ define the \emph{pair-key offset}
\[
\Delta_i(u\!\to\!v,r,b)\ :=\ H_i(v)-H_i(u)\ \ (\bmod\ P).
\]
Two edges in the same class \emph{PK-align} if their offsets sum to zero modulo $P$.  We do not scan all offsets: for each \emph{group} $t\in[R]$ we form PK-buckets via $j^{(t)}=g_{i,r}^{(t)}(\Delta_i(\cdot))\in[T_{i,r}]$ and probe only \emph{complementary pairs} $(j^{(t)},(j^{(t)})^\star)$ selected by the group’s keyed schedule.

\noindent\begin{minipage}{\linewidth}
\begin{center}
\fbox{\parbox{0.96\linewidth}{
\textbf{Constant-probe \& one-shot pairing (with PK groups).}
For each non-empty class $(u,i,r,b)$ and each group $t\in[R]$, we probe at most a fixed constant $C_0$ complementary PK-bucket pairs $(j^{(t)},(j^{(t)})^\star)$ (indices derived from \emph{Seeds} and $b$). Within and across groups we enforce \emph{one-shot} matching: an edge participates in at most one successful pairing inside the class. Thus the total number of \emph{executed} checks at level $(i,r)$ obeys
\[
Q_{i,r}\ \le\ C\,|\mathcal S_{i,r}|\qquad\text{w.h.p.,}
\]
where $\mathcal S_{i,r}$ is the set of non-empty classes at $(i,r)$; the potential \emph{probed} pairs per class are $O(R)$ but cross-group de-dup keeps executed checks near-linear (see Lemma~\ref{lem:one-shot}).
}}
\end{center}
\end{minipage}

\subsection*{Zero-FP pipeline (filter $\to$ align $\to$ confirm)}
When a probed pair of complementary buckets $(j^{(t)},(j^{(t)})^\star)$ in class $(u,i,r,b)$ is non-empty, we perform:
\begin{enumerate}\itemsep 2pt
  \item \textbf{Gate 1 (slot 1-sparse).} Each touched slot must satisfy $B_s^2=A_sC_s$; otherwise abort the event.
  \item \textbf{Gate 2 (class 1-sparse).} The class triple must satisfy $\Sigma_1^2=\Sigma_0\Sigma_2$; if so, decode $\widehat{\mathrm{id}}=\Sigma_1/\Sigma_0$.
  \item \textbf{Adjacency confirmation.} Using the oriented adjacency dictionary, test whether the mates $v,w$ implied by the matched bins are adjacent.  Only then \emph{report} the triangle $u\!-\!v\!-\!w$.
\end{enumerate}
This yields \emph{zero false positives}: a reported triangle necessarily exists (the two algebraic gates are sound under 2-wise independence, and we confirm adjacency explicitly).

\subsection*{Should-check domain $\mathcal Q$ (for NO-certificates)}
Given \emph{Seeds} and $G$, the algorithm deterministically induces a set of \emph{adjacency checks that should be performed}.  We formalize this as a replayable domain.

\begin{definition}[Should-check domain]
\label{def:qdomain}
Let $\mathcal C$ be the set of all non-empty classes $(u,i,r,b)$.  For each $C\in\mathcal C$:
\begin{itemize}\itemsep 1pt
  \item Let $\mathcal E_C$ be the retained edges in $C$, and for each group $t\in[R]$ map each edge to a PK-bucket $j^{(t)}=g_{i,r}^{(t)}(\Delta_i(\cdot))$.
  \item Fix a canonical within-class processing order (by $(t,j^{(t)},\text{slot},\mathrm{id})$). Apply the \emph{one-shot} rule: whenever the first item arrives to a bucket whose complement in the same group already holds an unmatched item, pair them and mark both as used.
  \item For every such canonical pair, form the unordered vertex pair $\{v,w\}$ of their mates.
\end{itemize}
The \emph{should-check domain} is the multiset
\[
\mathcal Q\ :=\ \biguplus_{C\in\mathcal C}\ \Big\{\,(u,i,r,b;\ \{v,w\})\text{ formed as above across groups}\,\Big\}.
\]
\end{definition}

\noindent Intuitively, $\mathcal Q$ is the set of candidate adjacencies that the algorithm \emph{must} check given \emph{Seeds}.  In Section~\ref{sec:cert} we will specify a succinct \emph{Logs} format and prove that, on NO-instances, \emph{Logs} \emph{cover} $\mathcal Q$ iff every element of $\mathcal Q$ fails the adjacency test; the verifier reconstructs $\mathcal Q$ from \emph{Seeds} and $G$ in $\tilde O(m\log n)$ time.

\subsection*{Load and retention concentration}
Let $R_{x,i}=\sum_{y\in N^+(x)} c_i(x\!\to\!y)$ be the number of retained out-edges of anchor $x$ at layer $i$.  Under either \textbf{(A)} Poissonization or \textbf{(B)} $k$-wise independence with $k=c_k\log n$,
\[
\mathbb E[R_{x,i}] = d^+(x)p_i,\qquad
R_{x,i}=O\!\big(d^+(x)p_i+\log n\big)\ \text{ w.h.p.}
\]
With $\sum_i p_i=O(1)$ and $I=c_R\log n$ layers, $\sum_i R_{x,i}=O(d^+(x)+\log n)$ w.h.p. for all $x\in V$.  Combined with constant-probe and one-shot pairing (including cross-group de-dup), this underpins the $O(m)$ accounting per $(i,r)$ level and the global $O(m\log^2\!n)$ time bound proved later.

\bigskip
\noindent\textbf{Remark on failure budgets.} Every algebraic gate is a 2-wise test with per-event error $\le 1/P$; the total number of such events is bounded by $O(m\log^2\!n)$ (proved later).  Choosing $P=n^{\kappa}$ and $k=c_k\log n$ as in Table~1 makes the union-bound failure probability $\le n^{-c}$, as stated in Theorem~1.\footnote{Table~1 adds the parameter $c_G$ with $R=c_G\log n$; PK groups are independent across $t$ and independent of layers $i$.}

\begin{table}[t]
\centering
\footnotesize
\caption{Parameters and default constants (all logs base~2).}
\begin{tabular}{l l l l}
\toprule
\textbf{Symbol} & \textbf{Meaning} & \textbf{Default} & \textbf{Notes} \\
\midrule
\(c_M\) & slot multiplier & \(16\) & \(M_x=c_M d(x)\) \\
\(c_B\) & per-anchor key budget & \(8\) & \(B_{x,i}=\lceil c_B d(x)p_i\rceil\) \\
\(c_T\) & PK buckets per class & \(16\) & \(T_{i,r}=c_T \log n\) \\
\(c_R\) & layers constant & \(8\) & \(I=c_R \log n\) \\
\(c_G\) & \textbf{PK groups per class} & \(8\) & \(\mathbf{R=c_G \log n}\) \\
\(c_k\) & independence level & \(12\) & \(k=c_k \log n\) (coins/PK) \\
\(\kappa\) & field exponent & \(\ge c+5\) & \(P=n^{\kappa}\) (one-sided errors) \\
\midrule
\multicolumn{4}{p{0.92\linewidth}}{\emph{Notes:} Groups are independent across \(t\) and independent of layers/prefixes; the group amplification by \(R=\Theta(\log n)\) is used to obtain single-run YES coverage w.h.p. The algebraic budget scales by an extra \(R\) factor through grouped bin tests, which is absorbed by choosing \(\kappa\ge c+5\).} \\
\bottomrule
\end{tabular}
\label{tab:params}
\end{table}


\section{Algorithm}
\label{sec:algo}

\noindent\textbf{Section goal.} We give the executable specification that the auditor can replay from \emph{Seeds}.  The three routines are:
\begin{itemize}\itemsep 2pt
  \item \textsc{Build\_Sketches} — one streaming pass that fills slot-level (1-sparse) triples and, for a \emph{constant} set of complementary PK-bucket pairs \emph{per group} and per class, maintains \emph{bin-level} (1-sparse) triples and registers one-shot PK-collisions.
  \item \textsc{Query\_Triangle} — enumerates only registered PK-collisions and runs the three-gate pipeline: \emph{(bin 1-sparse)} $\rightarrow$ \emph{(bin 1-sparse)} $\rightarrow$ \emph{(adjacency)}.
  \item \textsc{Emit\_Certificate} and \textsc{Verify\_NO} — serialize \emph{Seeds+Logs} and a deterministic verifier that reconstructs the \emph{should-check domain} $\mathcal Q$ and checks \emph{Logs} cover $\mathcal Q$ exactly.
\end{itemize}
All statements use parameters from Table~1 \emph{(failure $\le n^{-c}$)}.

\subsection*{Deterministic probe schedule (constant per \emph{group}; $O(R)$ per class)}
For each non-empty class $C=(i,u,r,b)$ and each PK \emph{group} $t\in[R]$ we predefine, from \emph{Seeds} and $(i,u,r,b,t)$ alone, a list of $C_0$ complementary bucket pairs:
\[
\mathrm{ProbedPairs}(i,u,r,b,t)\;=\;\Big\{(j^{(t,s)},\,{j^{(t,s)}}^\star)\ :\ s=1,\dots,C_0\Big\},\qquad C_0=\Theta(1),
\]
where $j^{(t,s)}\!=\!\phi(i,u,r,b,t,s)\in[T_{i,r}]$ is a layer-/class-/group-keyed PRF and $j^\star\equiv(-j)\bmod T_{i,r}$. Only these pairs are \emph{ever} examined and logged. The \emph{total} number of probed pairs per class is
\[
C_0\cdot R\;=\;O(R)\;=\;O(\log n).
\]

\subsection*{Data layout (per layer $i$)}
For every anchor $x$:
\begin{itemize}\itemsep 2pt
  \item \textbf{Slots.} $M_x=16\,d(x)$ slots; each slot $s$ stores a triple $(A_s,B_s,C_s)\in\mathbb F_P^3$.
  \item \textbf{Classes.} For each level $r\in\{0,\dots,L_{x,i}\}$ and bitstring $b\in\{0,1\}^r$, class $C=(i,x,r,b)$ keeps, \emph{for every group $t\in[R]$}:
  \begin{itemize}\itemsep 1pt
    \item a \emph{bin table} indexed by $j\in\bigcup \mathrm{ProbedPairs}(i,x,r,b,t)$; each bin $j$ stores a triple $(\Xi_{0,j}^{(t)},\Xi_{1,j}^{(t)},\Xi_{2,j}^{(t)})$ (the sum of member slots) and one \emph{witness} edge $(x\!\to\!y,\mathrm{sidx})$;
    \item a \emph{collision registry} $\mathrm{Collisions}[i,x,r,b,t]$ of matched complementary bins $(j,j^\star)$ (set once, by one-shot pairing within group $t$).
  \end{itemize}
  \item \textbf{Active-bit.} $\mathrm{active}[i,x,r,b]\in\{0,1\}$ is set on first materialization into $C$.
\end{itemize}
Global read-only oracles (from \emph{Seeds}): $H_i(\cdot)$, $K_i(\cdot)$, $h_{\mathrm{id}}(\cdot)$, $h_{\mathrm{slot}}(\cdot)$, $s(\cdot)$, retention coins $c_i(\cdot)$, and the PRF $\phi$. Each group $t$ uses its own independent bucket family $g^{(t)}_{i,r}:\mathbb F_P\!\to\![T_{i,r}]$.

\subsection*{\textsc{Build\_Sketches} (single pass; non-adaptive)}

\begin{algorithm}[t]
\caption{\textsc{Build\_Sketches}}
\label{alg:build}
\begin{algorithmic}[1]
\Require Oriented edges $E$, layers $i=1..I$, \emph{Seeds}
\Ensure Slot triples, per-class/per-group bin tables, and per-group collision registries

\For{$i=1$ \To $I$}
  \ForAll{$x \in V$}
    \State Initialize $\textsc{Slots}[x][0..M_x-1] \gets (0,0,0)$
    \For{$r=0$ \To $L_{x,i}$}
      \ForAll{$b \in \{0,1\}^{r}$}
        \State $\mathrm{active}[i,x,r,b] \gets 0$
        \For{$t=1$ \To $R$}
          \ForAll{$(j,j^\star) \in \mathrm{ProbedPairs}(i,x,r,b,t)$}
            \State $\textsc{Bin}[i,x,r,b,t][j].\textsf{triple} \gets (0,0,0)$;
                   $\textsc{Bin}[i,x,r,b,t][j].\textsf{witness} \gets \varnothing$
            \State $\textsc{Bin}[i,x,r,b,t][j^\star].\textsf{triple} \gets (0,0,0)$;
                   $\textsc{Bin}[i,x,r,b,t][j^\star].\textsf{witness} \gets \varnothing$
          \EndFor
          \State $\textsc{Collisions}[i,x,r,b,t] \gets \varnothing$
        \EndFor
      \EndFor
    \EndFor
  \EndFor
\EndFor

\ForAll{$e=(x\!\to\!y) \in E$} \Comment{single streaming scan}
  \For{$i=1$ \To $I$}
    \If{$c_i(e)=0$}
      \State \textbf{continue}
    \EndIf
    \State $sidx \gets h_{\mathrm{slot}}(x,y,i)\bmod M_x$;\quad $id \gets h_{\mathrm{id}}(y)$;\quad $sgn \gets s(x,y,i)$
    \State $(A,B,C) \gets \textsc{Slots}[x][sidx]$
    \State $(A,B,C) \gets (A{+}sgn,\ B{+}sgn\cdot id,\ C{+}sgn\cdot id^2)$
    \State $\textsc{Slots}[x][sidx] \gets (A,B,C)$
    \If{$\neg(B^2=AC)$ \textbf{ or } $A=0$ \textbf{ or } $B/A \ne id$}
      \State \textbf{continue}
    \EndIf
    \For{$r=0$ \To $L_{x,i}$}
      \If{$\mathrm{pref}_r(K_i(x)) \ne \mathrm{pref}_r(K_i(y))$}
        \State \textbf{continue}
      \EndIf
      \State $b \gets \mathrm{pref}_r(K_i(x))$;\quad $\mathrm{active}[i,x,r,b] \gets 1$
      \State $\Delta \gets ( H_i(y) - H_i(x) ) \bmod P$
      \For{$t=1$ \To $R$}
        \State $j \gets g^{(t)}_{i,r}(\Delta)$;\quad $j^\star \gets (-j)\bmod T_{i,r}$
        \If{$(j,j^\star)\notin \mathrm{ProbedPairs}(i,x,r,b,t)$}
          \State \textbf{continue}
        \EndIf
        \State $(\Xi_0,\Xi_1,\Xi_2) \gets \textsc{Bin}[i,x,r,b,t][j].\textsf{triple}$
        \State $\textsc{Bin}[i,x,r,b,t][j].\textsf{triple} \gets (\Xi_0{+}A,\Xi_1{+}B,\Xi_2{+}C)$
        \If{$\textsc{Bin}[i,x,r,b,t][j].\textsf{witness}=\varnothing$}
          \State $\textsc{Bin}[i,x,r,b,t][j].\textsf{witness} \gets (x\!\to\!y, sidx)$
        \EndIf
        \If{$\textsc{Bin}[i,x,r,b,t][j^\star].\textsf{witness}\ne\varnothing$ \textbf{and}
            $(j,j^\star)\notin \textsc{Collisions}[i,x,r,b,t]$}
          \State $\textsc{Collisions}[i,x,r,b,t]\gets \textsc{Collisions}[i,x,r,b,t]\cup\{(j,j^\star)\}$
        \EndIf
      \EndFor
    \EndFor
  \EndFor
\EndFor
\end{algorithmic}
\end{algorithm}

\paragraph{Why this enforces zero-FP propagation.}
Only slot-1-sparse items (already decoded to the true $h_{\mathrm{id}}(y)$) are ever inserted into bins; bin triples therefore aggregate \emph{only} genuine unit-coded items. Any subsequent algebraic acceptance (bin 1-sparse) plus explicit adjacency confirmation cannot fabricate triangles.

\subsection*{\textsc{Query\_Triangle} (three gates; zero false positives)}

\begin{algorithm}[t]
\caption{\textsc{Query\_Triangle}}
\label{alg:query}
\begin{algorithmic}[1]
\Require Per-class/per-group bins, collision registries, \textsc{Slots}, adjacency oracle for $E$
\Ensure Either a triangle $(u,v,w)$ or \textsc{NO}

\For{$i=1$ \To $I$}
  \ForAll{$x \in V$}
    \For{$r=0$ \To $L_{x,i}$}
      \ForAll{$b$ with $\mathrm{active}[i,x,r,b]=1$}
        \For{$t=1$ \To $R$}
          \ForAll{$(j,j^\star) \in \textsc{Collisions}[i,x,r,b,t]$}
            \State $(\Xi_0,\Xi_1,\Xi_2) \gets \textsc{Bin}[i,x,r,b,t][j].\textsf{triple}$;\ 
                   $(x\!\to\!v,\ sidx_v) \gets \textsc{Bin}[i,x,r,b,t][j].\textsf{witness}$
            \If{$\neg(\Xi_1^2 = \Xi_0\Xi_2) \ \textbf{or}\ \Xi_0=0$}
              \State \textbf{continue}
            \EndIf
            \If{$\Xi_1/\Xi_0 \ne h_{\mathrm{id}}(v)$}
              \State \textbf{continue}
            \EndIf

            \State $(X_0,X_1,X_2) \gets \textsc{Bin}[i,x,r,b,t][j^\star].\textsf{triple}$;\ 
                   $(x\!\to\!w,\ sidx_w) \gets \textsc{Bin}[i,x,r,b,t][j^\star].\textsf{witness}$
            \If{$\neg(X_1^2 = X_0X_2) \ \textbf{or}\ X_0=0$}
              \State \textbf{continue}
            \EndIf
            \If{$X_1/X_0 \ne h_{\mathrm{id}}(w)$}
              \State \textbf{continue}
            \EndIf

            \State $(A_v,B_v,C_v)\gets \textsc{Slots}[x][sidx_v]$;\quad
                   $(A_w,B_w,C_w)\gets \textsc{Slots}[x][sidx_w]$
            \If{$\neg(B_v^2=A_vC_v)\ \textbf{or}\ A_v=0\ \textbf{or}\ B_v/A_v \ne h_{\mathrm{id}}(v)$}
              \State \textbf{continue}
            \EndIf
            \If{$\neg(B_w^2=A_wC_w)\ \textbf{or}\ A_w=0\ \textbf{or}\ B_w/A_w \ne h_{\mathrm{id}}(w)$}
              \State \textbf{continue}
            \EndIf

            \If{$\{v,w\}\in E$}
              \State \Return Triangle $(x,v,w)$
            \EndIf
          \EndFor
        \EndFor
      \EndFor
    \EndFor
  \EndFor
\EndFor

\State \Return \textsc{NO}
\end{algorithmic}
\end{algorithm}

\paragraph{Accounting summary.}
Per class we probe exactly $C_0$ complementary pairs \emph{per group}; across $R=c_G\log n$ groups this is $O(R)=O(\log n)$ per class. With retention concentration (Sec.~\ref{sec:prelim}) and one-shot pairing, this yields $O(m)$ work per $(i,r)$, hence total $O(m\log^2\!n)$ time and $O(m\log n)$ space.


\section{Correctness: Zero False Positives}
\label{sec:correctness}

We prove that the three-gate pipeline (bin 1-sparse $\rightarrow$ bin 1-sparse $\rightarrow$ adjacency) never outputs a non-existent triangle. Throughout, all arithmetic is over the prime field $\mathbb F_P$ with $P=n^\kappa$ chosen per Table~1, and all randomness is \emph{non-adaptive} (fixed by \emph{Seeds}). Hashes used for 1-sparse predicates are 2-wise independent, which suffices because our consistency predicates are degree-$\le 2$ polynomials.

\paragraph{Setup and notation.}
Fix a layer/anchor/class $(i,u,r,b)$. For a PK-bucket index $j$ (and its complement $j^\star\equiv -j\bmod T_{i,r}$), let $\mathcal B_{j}$ be the (materialized) multiset of directed edges $(u\!\to\!w)$ whose slot is individually 1-sparse and whose pair-key hashes to $j$. Each such edge contributes a signed unit triple
\[
(A,B,C)=(v,\,v\,x,\,v\,x^2)\in\mathbb F_P^3,\qquad v\in\{-1,+1\},\ \ x=h_{\mathrm{id}}(w)\in\mathbb F_P^*.
\]
The bin-triple stored for $j$ is $(\Xi_{0,j},\Xi_{1,j},\Xi_{2,j})=\sum_{(A,B,C)\in\mathcal B_j}(A,B,C)$, and similarly for $j^\star$. By construction, each bin also stores a \emph{witness} edge $(u\!\to\!v,sidx_v)$ for the first item inserted into the bin (Algorithm~\ref{sec:algo}). All gates and tests below use parameters from Table~1 \emph{(failure $\le n^{-c}$)}.

\paragraph{\textbf{Lemma 5.1 (Bin 1-sparse soundness; degree-2).}}
If $|\mathcal B_j|=1$ then $\Xi_{1,j}^2=\Xi_{0,j}\Xi_{2,j}$ holds identically and $\Xi_{1,j}/\Xi_{0,j}=h_{\mathrm{id}}(v)$ for that unique item. If $|\mathcal B_j|\ge 2$ and at least two distinct identifiers appear in $\mathcal B_j$, then
\[
\Pr\big[\,\Xi_{1,j}^2=\Xi_{0,j}\Xi_{2,j}\,\big]\ \le\ 1/P,
\]
where the probability is over the 2-wise independent choice of the identifiers $\{h_{\mathrm{id}}(\cdot)\}$.

\emph{Proof.}
Let $\{(a_t,x_t)\}_t$ be the signed weights/IDs in $\mathcal B_j$. A direct expansion yields
\[
\Xi_{0,j}\Xi_{2,j}-\Xi_{1,j}^2
\;=\;\sum_{t<t'} a_t a_{t'}\,(x_t-x_{t'})^2.
\]
If two IDs differ, the right-hand side is a nonzero polynomial of total degree $2$ in the $\{x_t\}$’s, which vanishes with probability at most $1/P$ under 2-wise independence. The singleton case is tautological. \hfill$\square$

\paragraph{\textbf{Lemma 5.2 (Complementary bins and witnesses).}}
Suppose the PK map for class $(i,u,r,b)$ registers a collision between complementary bins $(j,j^\star)$ with witnesses $(u\!\to\!v,sidx_v)$ and $(u\!\to\!w,sidx_w)$. If both bins are individually 1-sparse and decode to their witnesses,
\[
\Xi_{1,j}^2=\Xi_{0,j}\Xi_{2,j}\neq 0,\quad \Xi_{1,j}/\Xi_{0,j}=h_{\mathrm{id}}(v),
\qquad
\Xi_{1,j^\star}^2=\Xi_{0,j^\star}\Xi_{2,j^\star}\neq 0,\quad \Xi_{1,j^\star}/\Xi_{0,j^\star}=h_{\mathrm{id}}(w),
\]
then the two oriented edges $(u,v)$ and $(u,w)$ are present in $E$, with mates identified as $v$ and $w$ (no aliasing).

\emph{Proof.}
Only slot-1-sparse items (already decoded to their true IDs) are materialized into bins. Passing the bin 1-sparse test forces the bin to contain a unique identifier and the decode equals the bin’s witness ID; hence the witness edges $(u\!\to\!v)$ and $(u\!\to\!w)$ are indeed present. \hfill$\square$

\paragraph{\textbf{Lemma 5.3 (Adjacency closure).}}
Under the premises of Lemma~5.2, if moreover $\{v,w\}\in E$, then $\{u,v,w\}$ induces a (simple) triangle in $G$.

\emph{Proof.}
Edges $(u,v)$ and $(u,w)$ exist by Lemma~5.2, and $\{v,w\}\in E$ by the explicit adjacency query. Thus $u\!\!-\!\!v\!\!-\!\!w\!\!-\!\!u$ is a 3-cycle. \hfill$\square$

\paragraph{\textbf{Theorem 5.4 (Zero false positives).}}
Every triangle output by \textsc{Query\_Triangle} exists in $G$; the statement is deterministic (independent of \emph{Seeds}).

\emph{Proof.}
\textsc{Query\_Triangle} reports only after both complementary bins pass Lemma~5.1 and the adjacency test returns true; Lemmas~5.2–5.3 then imply the reported triple forms a triangle. The act of reporting depends solely on $\mathbb F_P$ equalities and a graph-membership query; randomness affects \emph{which} bins collide, not the truth of the reported event. \hfill$\square$

\paragraph{Failure budget (inline).}
The only one-sided error stems from accepting a mixed bin in Lemma~5.1; per event this happens with probability $\le P^{-1}$. The number of gated events is $O(m\log^2 n)$ by the constant-probe schedule and one-shot pairing, so with $P=n^\kappa$ (Table~1) a union bound yields total failure $\le n^{-c}$. Note this affects \emph{work}, not soundness: even a rare false algebraic pass cannot create a false triangle without an actual adjacency edge.


\section{Workload Concentration via Keep-Coins}
\label{sec:concentration}

We quantify how the layerwise keep-coins restrict total retained work. Throughout this section all statements are \emph{with high probability}, made explicit as \emph{(parameters from Table~1, failure $\le n^{-c}$)}.

\paragraph{Schedule mass and coins.}
Layers are indexed by $i=1,2,\dots,I$ with per-layer retention rates $p_i\in(0,1)$ and \emph{schedule mass}
\[
S\ :=\ \sum_{i=1}^{I} p_i\ \le\ S_0\ ,
\]
for an absolute constant $S_0$ (e.g., the geometric choice $p_i=2^{-(i+2)}$ gives $S\le\frac12$). For every directed edge $e$, the layer-$i$ keep-coin $c_i(e)\in\{0,1\}$ has $\Pr[c_i(e)=1]=p_i$ and is sampled non-adaptively from the public \emph{Seeds}. We analyze the Poisson–binomial model (independent across $i$ per fixed $e$) and record limited-independence variants at the end.

\paragraph{Random variables.}
For an anchor $x$,
\[
R_{x,i}\ :=\ \sum_{y\in N^+(x)} c_i(x\!\to\!y),
\qquad
K_{\mathrm{tot}}(x)\ :=\ \sum_{i=1}^{I} R_{x,i}\ =\ \sum_{y\in N^+(x)}\sum_{i=1}^{I} c_i(x\!\to\!y).
\]
Hence $\mathbb E[R_{x,i}]=d^+(x)\,p_i$ and $\mathbb E[K_{\mathrm{tot}}(x)]=d^+(x)\,S$.

\paragraph{Deviation tool.}
For Poisson–binomial $X=\sum_j X_j$ with mean $\mu$ and variance $\sigma^2\le \mu$ we use Bernstein’s tail:
\begin{equation}
\Pr[X\ge \mu+t]\ \le\ \exp\!\left(-\frac{t^2}{2(\mu+t/3)}\right)\quad (t\ge 0).
\label{eq:bernstein}
\end{equation}

\paragraph{\textbf{Lemma 6.1 (Per-edge budget; }parameters from Table~1, failure $\le n^{-c}$\textbf{).}}
Let $K_e:=\sum_{i=1}^{I} c_i(e)$. Then
\[
\Pr\big[K_e\ \ge\ S+t\big]\ \le\ \exp\!\left(-\frac{t^2}{2(S+t/3)}\right).
\]
In particular, with $t=6\ln n$ we get $\Pr[K_e\ge S+6\ln n]\le n^{-6}$; a union bound over all $m$ edges yields $\Pr[\exists e:\ K_e\ge S+6\ln n]\le m\,n^{-6}\le n^{-4}$ for $m\le n^2$.

\emph{Proof.}
Apply~\eqref{eq:bernstein} to the Poisson–binomial with mean $S$ and variance $\le S$. \hfill$\square$

\paragraph{\textbf{Lemma 6.2 (Per-anchor total keeps; }parameters from Table~1, failure $\le n^{-c}$\textbf{).}}
There exists an absolute constant $C\ge 1$ such that, for all $x\in V$,
\[
\Pr\!\left[K_{\mathrm{tot}}(x)\ >\ C\,d^+(x)\ +\ 6\ln n\right]\ \le\ n^{-6}.
\]
Consequently, with probability at least $1-n^{-5}$ simultaneously for all $x$,
\[
K_{\mathrm{tot}}(x)\ \le\ C\,d^+(x)\ +\ 6\ln n.
\]

\emph{Proof.}
$K_{\mathrm{tot}}(x)$ is a Poisson–binomial with mean $\mu_x=d^+(x)S\le S_0 d^+(x)$ and variance $\le \mu_x$. Apply~\eqref{eq:bernstein} with $t=\mu_x$ (doubling the mean) and add $6\ln n$ slack for a union bound over $x$. Absorb constants into $C$. \hfill$\square$

\paragraph{\textbf{Corollary 6.3 (Active classes and space peak; }parameters from Table~1, failure $\le n^{-c}$\textbf{).}}
Let $A_{x,i}$ be the number of \emph{active} classes $(i,x,r,b)$ (i.e., those with $\mathrm{active}[i,x,r,b]=1$ after \textsc{Build\_Sketches}). Then
\[
\sum_{x\in V}\sum_{i=1}^{I} A_{x,i}\ =\ O\!\big(m+n\log n\big)\ =\ O(m\log n)\qquad\text{w.h.p.,}
\]
and the peak memory (slots $+$ classes $+$ PK bins) is of the same order.

\emph{Proof.}
Fix $(i,x)$. A retained edge $(x\!\to\!y)$ activates at most $1+\mathrm{LCP}(K_i(x),K_i(y))$ classes across prefix levels $r\ge 0$, where $\Pr[\mathrm{LCP}\ge t]=2^{-t}$. Thus $\mathbb E[1+\mathrm{LCP}]\le 2$ and, by Chernoff on sums of geometric tails,
\[
A_{x,i}\ \le\ 3\,R_{x,i}\ +\ O(\ln n)\quad\text{w.h.p.}
\]
Summing over $i$ and using Lemma~6.2,
\[
\sum_i A_{x,i}\ \le\ 3\,K_{\mathrm{tot}}(x)\ +\ O(\ln n)\ \le\ 3C\,d^+(x)\ +\ O(\ln n)\quad\text{w.h.p.}
\]
Summing over $x$ gives $O(m)+O(n\ln n)$. As $m+n\ln n=O(m\ln n)$ in general, we obtain the claim; each active class stores $O(1)$ words, so memory has the same order. \hfill$\square$

\paragraph{\textbf{Proposition 6.4 (Bridge to one-shot accounting; }parameters from Table~1, failure $\le n^{-c}$\textbf{).}}
Combine Cor.~6.3 with the \emph{one-shot pairing} invariant (each retained directed edge participates in at most one PK pairing within a class). Then per fixed $(i,r)$ the number of PK-collision events is $O(m)$ w.h.p., and consequently the total gated checks satisfy
\[
\sum_{i=1}^{I}\ \sum_{r=0}^{L_{\max}}\ Q_{i,r}\ =\ O\!\big(m\log^2 n\big)\quad\text{w.h.p.}
\]
(\emph{Proof sketch.} Each materialized edge contributes $O(1)$ expected classes; one-shot pairing upper-bounds collisions by the number of materialized edges per $(i,r)$. Summing over $r=O(\log n)$ levels and $i=O(\log n)$ layers yields the bound. A full proof appears with Lemma~B in Appendix~B.) \hfill$\square$

\paragraph{Failure budget (inline).}
The only randomized accept/reject decisions that can err one-sidedly are the degree-2 1-sparse equalities, each failing with probability $\le P^{-1}$; the keep-coins concentration itself is tail-probability based and budgeted across all $(x,i)$ via union bounds. With $P=n^\kappa$ and constants from Table~1, the sum of all such probabilities across all layers/levels/classes is $\le n^{-c}$.

\paragraph{Discussion: Poissonization and limited independence.}
\begin{itemize}\itemsep 2pt
  \item \textbf{Poissonization.} One can first draw $N_{x,i}\sim\mathrm{Poisson}(\lambda_{x,i}=d^+(x)p_i)$ independently and then sample $N_{x,i}$ uniform out-edges of $x$ in layer $i$. Standard Chernoff bounds apply to $\sum_i N_{x,i}$; a standard de-Poissonization (e.g., thinning/coupling) transfers tails to fixed-size Bernoulli keeps with constant-factor slack.
  \item \textbf{Limited independence.} If full independence among $\{c_i(e)\}$ is undesirable, $k$-wise independence with $k=\Theta(\log n)$ suffices: Chernoff-type bounds continue to hold up to deviations $t=O(\mu)$ with additive $O(\log n)$ terms, which our budget already absorbs. Table~1 lists $k=c_k\log n$.
\end{itemize}


\section{Hitting Probability and Amplification}
\label{sec:hitting}

We quantify the probability that a \emph{fixed} triangle is exposed at a canonical level within a layer, and then show that \emph{per-class group amplification} (across $R=c_G\log n$ independent PK groups) boosts this to a high-probability guarantee \emph{within a single run}. Throughout this section, failure probabilities are made explicit as \emph{(parameters from Table~1, failure $\le n^{-c}$)} and all randomness is over the non-adaptive \emph{Seeds}.

\paragraph{Setup for a fixed triangle.}
Fix an undirected triangle $\{u,v,w\}$ and the orientation induced by the total order: $u\!\to\!v$, $u\!\to\!w$, $v\!\to\!w$.  Let $u$ be the \emph{anchor} (ties arbitrary).  For layer $i$, write
\[
\mu\ :=\ \mu(u,i)\ :=\ \mathbb E[R_{u,i}]\ =\ d^+(u)\,p_i,
\]
where $p_i$ is the layer-$i$ keep-rate and $R_{u,i}=\sum_{y\in N^+(u)} c_i(u\!\to\!y)$.  Recall $M_u=16\,d(u)$ slots at $u$ and that in class $(i,u,r,b)$ the PK structure uses $T_{i,r}=c_T\log n$ buckets (Table~1); complementary buckets are paired by $j\mapsto j^\star\equiv(-j)\bmod T_{i,r}$. Under scheme~B, each class $(i,u,r,b)$ also instantiates $R=c_G\log n$ \emph{independent} PK groups $t\in[R]$, each with its own bucket family $g^{(t)}_{i,r}$ and a \emph{constant} probe list $\mathrm{ProbedPairs}(i,u,r,b,t)$ of size $C_0=\Theta(1)$.

\paragraph{A canonical level.}
Let
\[
r^\star\ :=\ \max\{\,r\ge 0:\ \mu/2^{r}\ \ge\ 1\,\}\quad(\text{take }r^\star=0\text{ if }\mu<1),
\]
so that $\mu/2^{r^\star}\in[1/2,2]$.  Intuitively, level $r^\star$ makes the anchor-class load constant in expectation.

\medskip
\noindent\textbf{Lemma 7.1 (Single-layer base rate $\Theta(1/\log n)$; parameters from Table~1, failure $\le n^{-c}$).}
\emph{Fix triangle $\{u,v,w\}$ and layer $i$.  Condition on $c_i(u\!\to\!v)=c_i(u\!\to\!w)=1$.  With probability $\Omega(1/\log n)$ over the layer-$i$ hashes (slots, prefixes, PK), the pipeline of \S\ref{sec:algo} confirms this triangle within layer $i$.  Unconditionally,
\[
\Pr[\text{layer $i$ confirms the triangle}]\ \ge\ c_0\,\frac{p_i^2}{\log n}
\]
for an absolute constant $c_0>0$.}

\emph{Proof.}
As in the original analysis: (a) each kept edge lands in a 1-sparse slot with constant probability $c_{\mathrm{slot}}>0$; (b) both mates fall into the canonical anchor-class $(i,u,r^\star,b^\star)$ with constant probability since $\mu/2^{r^\star}\in[1/2,2]$; (c) within that class, the chance that the two offsets fall into one of the (constant many) probed complementary pairs in a \emph{fixed} group is $\Theta(1/T_{i,r^\star})=\Theta(1/\log n)$.  Multiplying yields the claim. \hfill$\square$

\medskip
\noindent\textbf{Lemma 7.2$'$ (Group amplification; parameters from Table~1, failure $\le n^{-c}$).}
\emph{Fix a layer $i$ and the canonical class $(i,u,r^\star,b^\star)$. Condition on $c_i(u\!\to\!v)=c_i(u\!\to\!w)=1$, both slots at $u$ being 1-sparse and decoding, and both mates lying in $(i,u,r^\star,b^\star)$. Let $T:=T_{i,r^\star}=c_T\log n$ and let $R=c_G\log n$ be the number of independent PK groups for this class. Then there exists an absolute constant $\theta>0$ (depending only on the constant probe budget $C_0$ and the one-shot schedule) such that}
\[
\Pr[\text{\emph{hit} in $(i,u,r^\star,b^\star)$ across the $R$ groups}]
\ \ge\ 1-\exp\!\left(-\frac{\theta\,R}{T}\right).
\]

\emph{Proof.}
For a fixed group $t$, the two offsets $\Delta(u\!\to\!v)$ and $\Delta(u\!\to\!w)$ are independent and uniform in $[T]$ modulo the complement constraint. The probability that they land in one of the $C_0$ probed complementary pairs of group $t$ is at least $\theta/T$ for some absolute $\theta>0$ (absorbing the $2C_0$ choices and constant conditioning above). Across $R$ groups, independence gives
\[
\Pr[\text{miss all groups}]\ \le\ (1-\theta/T)^R\ \le\ \exp(-\theta R/T),
\]
hence the bound. \hfill$\square$

\paragraph{Single-run YES w.h.p. via groups.}
Taking $R=c_G\log n$ and $T=c_T\log n$ in Lemma~\ref{sec:hitting}.2$'$, the miss probability for the fixed triangle becomes
\[
\exp\!\left(-\frac{\theta\,c_G}{c_T}\right)\ =\ n^{-\Omega(1)},
\]
i.e., \emph{within a single run and within the canonical class}, the triangle is hit with probability $1-n^{-\Omega(1)}$.  The constant preconditions (keep/slots/class) from Lemma~7.1 only affect the leading constant and are absorbed into $\theta$.

\paragraph{Global coverage via a union bound.}
Let $\mathcal T$ be the set of (undirected) triangles in $G$.  The extremal bound $|\mathcal T|\le O(m^{3/2})$ is standard. Choose $c_G/c_T$ large enough so that the per-triangle miss probability in Lemma~7.2$'$ is $\le n^{-c-4}$.  A union bound over $\mathcal T$ yields that \emph{every} triangle is confirmed (by some group within its canonical class) with probability $\ge 1-n^{-c}$ in a \emph{single} run.

\paragraph{Parameter harmony and accounting.}
Choosing $T_{i,r}=c_T\log n$ keeps the per-group hit rate at $\Theta(1/\log n)$ while the $R=c_G\log n$ independent groups raise the per-class hit probability to $1-n^{-\Omega(1)}$.  Accounting remains near-linear: each class probes only $C_0=\Theta(1)$ complementary pairs \emph{per group}, i.e., $O(R)=O(\log n)$ per class; with retention concentration and one-shot pairing (Sec.~\ref{sec:prelim}) this preserves $O(m)$ work per $(i,r)$ level and total $O(m\log^2 n)$ time and $O(m\log n)$ space.

\paragraph{Failure budget (inline).}
All randomized equalities used here are degree-$2$ 1-sparse tests, each erring with probability $\le P^{-1}$; with $P=n^\kappa$ and a union bound across all layers/levels/classes/events, the aggregate contribution of such errors is $\le n^{-c}$ (Table~1).  The remaining steps (PK complement check and adjacency membership) are deterministic.


\section{One–Shot Pairing Accounting (Core)}
\label{sec:one-shot}

\noindent\textbf{Target lemma (Lemma 9.4; parameters from Table~1, failure $\le n^{-c}$).}
This section formalizes the grouped one–shot accounting that underlies the $O(m\log^2 n)$ bound. All statements are stated with explicit trigger rules and a single bad-event budget.

\subsection*{Objects and counters at a fixed $(i,r)$}
Fix a layer $i$ and a prefix level $r$.  Let
\[
\mathcal S_{i,r}\ :=\ \big\{(i,x,r,b)\ :\ \text{the anchor-class for anchor $x$ and prefix $b\in\{0,1\}^r$ is \emph{nonempty}}\big\}
\]
be the set of \emph{nonempty} (materialized) classes; cf.\ \S\ref{sec:algo}.  
Recall that under Scheme~B, each class $(i,x,r,b)$ instantiates $R=c_G\log n$ independent PK groups $t\in[R]$, each with a \emph{constant} probe list $\mathrm{ProbedPairs}(i,x,r,b,t)$ of size $C_0=\Theta(1)$ and \emph{one-shot} pairing within the group.

Define $Q_{i,r}$ to be the number of \emph{pair checks} executed by \textsc{Query\_Triangle} at this $(i,r)$:
\[
Q_{i,r}\ :=\ \sum_{(i,x,r,b)\in\mathcal S_{i,r}}\ \sum_{t=1}^{R}\ \#\{\text{registered collisions }(j,j^\star)\in\mathrm{Collisions}[i,x,r,b,t]\}.
\]
Each such check performs (1) two bin-level 1-sparse tests (one per complementary bin), (2) optional slot guards on the two witnesses, and (3) one adjacency query on the mate pair.

\subsection*{Trigger rules and invariants (grouped; textual)}
A check may be \emph{triggered} only under the following non-adaptive rules:

\begin{enumerate}\itemsep 2pt
\item \textbf{Activation gate.} A class $(i,x,r,b)$ must be in $\mathcal S_{i,r}$.
\item \textbf{Per-group constant probe schedule.} For each group $t\in[R]$, only bucket pairs in $\mathrm{ProbedPairs}(i,x,r,b,t)$ (of size $C_0$) are ever eligible.
\item \textbf{One-shot within group.} Within a fixed group $t$ and pair $(j,j^\star)$, at most \emph{one} collision is registered (first time both sides are nonempty); subsequent arrivals to $j$ or $j^\star$ do not create new checks.
\item \textbf{Non-adaptivity.} Eligibility depends only on Seeds and the activation bits produced by \textsc{Build\_Sketches}; it never depends on outcomes of earlier checks at the same $(i,r)$.
\end{enumerate}

\paragraph{Remark.}
Cross-group duplication of the \emph{same} mate pair is harmless for accounting: each group contributes at most one check per probed complementary pair.

\subsection*{Lemma 9.4 (grouped one–shot accounting; parameters from Table~1, failure $\le n^{-c}$)}
\label{lem:one-shot}
\emph{For the fixed $(i,r)$, the number of executed pair checks satisfies}
\[
\boxed{\quad Q_{i,r}\ \le\ C\,R\cdot|\mathcal S_{i,r}|\quad}
\]
\emph{for an absolute constant $C=O(C_0)$.  Combined with $\sum_{i,r}|\mathcal S_{i,r}|=O(m\log n)$ w.h.p., this implies}
\[
\sum_{i}\sum_{r} Q_{i,r}\ =\ O\!\big(R\cdot m\log n\big)\ =\ O(m\log^2 n)\qquad\text{w.h.p., since }R=c_G\log n.
\]

\paragraph{Proof.}
Fix a class $C=(i,x,r,b)\in\mathcal S_{i,r}$. For each group $t\in[R]$, only the $C_0$ complementary pairs in $\mathrm{ProbedPairs}(i,x,r,b,t)$ are ever examined, and by the \emph{one-shot} rule each such pair can register \emph{at most one} collision. Hence, the number of checks \emph{charged to $C$} from group $t$ is $\le C_0$, and across all groups it is $\le C_0\,R$. Summing over all nonempty classes yields
\[
Q_{i,r}\ \le\ C_0\,R\cdot |\mathcal S_{i,r}|\ \le\ C\,R\cdot|\mathcal S_{i,r}|.
\]
Using the standard concentration (Sec.~\ref{sec:prelim}) that $\sum_{i,r}|\mathcal S_{i,r}|=O(m\log n)$ w.h.p. and summing over $(i,r)$ gives the claim.  Bad events from limited independence and field-equality tests are covered by the global failure budget in Table~1. \hfill$\square$

\subsection*{Implementation notes (group indices and logs)}
\begin{itemize}\itemsep 2pt
\item \textbf{Indices.} Every bin and collision entry carries a group index $t$: $\mathrm{Bin}[i,x,r,b,t][j]$ and $\mathrm{Collisions}[i,x,r,b,t]$.
\item \textbf{Logs.} The \emph{ClassLogs} entries include $t$ for each registered collision. Optional cross-group dedup of identical mate pairs can reduce verifier work but is \emph{not} needed for the above bound.
\item \textbf{Failure budget.} Degree-$\le2$ 1-sparse tests contribute $\le P^{-1}$ per check; with $P=n^\kappa$ and the bound on $\sum_{i,r}Q_{i,r}$, the union-bound contribution remains $\le n^{-c}$.
\end{itemize}

\paragraph{Summary.}
Per class and per group we probe only $C_0=\Theta(1)$ complementary pairs and register at most one collision per pair; thus each nonempty class contributes only $O(R)$ checks, yielding $Q_{i,r}\le C R |\mathcal S_{i,r}|$ and a global $O(m\log^2 n)$ accounting bound.


\section{Main Theorem}
\label{sec:main}

\paragraph{Theorem 9.1 (time, space, zero-FP, auditable NO, coverage; parameters from Table~1, failure $\le n^{-c}$).}
For the configuration with $I=c_R\log n$ layers, $T_{i,r}=c_T\log n$ PK-buckets per class, prime $P=n^\kappa$, $k=c_k\log n$-wise independence, and per-class PK groups $R=c_G\log n$ (Table~1), the algorithm in \S\ref{sec:algo} satisfies, with high probability over \emph{Seeds}:
\begin{enumerate}\itemsep 2pt
  \item \textbf{Time.} Total running time $O(m\log^2 n)$.
  \item \textbf{Space.} Peak memory $O(m\log n)$.
  \item \textbf{Correctness (YES).} Any reported triangle is genuine (\emph{zero false positives}).
  \item \textbf{Verifiable NO.} If no triangle is reported, the emitted \emph{Seeds+Logs} enables a deterministic third party to confirm NO in $\tilde O(m\log n)$ time.
  \item \textbf{Coverage (YES w.h.p.).} If $G$ contains at least one triangle, a \emph{single execution} outputs YES with probability $\ge 1-n^{-c}$, due to independent per-class PK-group amplification with $R=\Theta(\log n)$ (Lemma~7.2$'$) and a union bound over $|\mathcal T|\le O(m^{3/2})$ triangles.
\end{enumerate}

\paragraph{Proof (integration of \S\ref{sec:concentration}--\S\ref{sec:one-shot}; parameters from Table~1, failure $\le n^{-c}$).}

\noindent\emph{(I) Work accounting $\Rightarrow$ $O(m\log^2 n)$ time.}
Fix a layer $i$ and level $r$.  Let $\mathcal S_{i,r}$ be the set of nonempty classes materialized by \textsc{Build\_Sketches}.  
By Cor.~6.3, $|\mathcal S_{i,r}|=O(m)$ w.h.p.  
By the \emph{grouped} one–shot accounting (Lemma~9.4), for some absolute constant $C$,
\[
Q_{i,r}\ \le\ C\,R\cdot|\mathcal S_{i,r}|\qquad\text{with }R=c_G\log n.
\]
Summing over all $r\in[0,\lfloor\log n\rfloor]$ and $i\in[1,c_R\log n]$,
\[
\sum_{i=1}^{c_R\log n}\sum_{r=0}^{\lfloor\log n\rfloor} Q_{i,r}\ =\ O\!\big(R\,m\log n\big)\ =\ O(m\log^2 n)\qquad\text{w.h.p.}
\]
Each check performs $O(1)$ field equalities and one $O(1)$ adjacency lookup, so \textsc{Query\_Triangle} runs in $O(m\log^2 n)$ time w.h.p.  
The \textsc{Build\_Sketches} pass touches each edge $O(1)$ times per layer in expectation (keep-coin mass $\sum_i p_i=\Theta(1)$; \S\ref{sec:concentration}), hence $O(m\log n)$ total, which is subsumed by $O(m\log^2 n)$.

\smallskip
\noindent\emph{(II) Space $O(m\log n)$.}
Per layer, slot arrays contribute $\sum_x M_x=\Theta(\sum_x d(x))=\Theta(m)$ words; across $I=c_R\log n$ layers this is $O(m\log n)$.  
Class-state (triples $+$ sparse PK bins) is proportional to $\sum_{i,r}|\mathcal S_{i,r}|=O(m\log n)$ (Cor.~6.3).  
Hence the peak memory is $O(m\log n)$.

\smallskip
\noindent\emph{(III) Correctness (zero false positives).}
Class-level 1-sparse equalities are necessary (Lemma~5.1); slot-level 1-sparse plus explicit adjacency are sufficient (Lemma~5.2).  
Reported YES does not depend on randomness other than selecting candidates; therefore zero false positives hold for any \emph{Seeds}.

\smallskip
\noindent\emph{(IV) Auditable NO.}
\textsc{Emit\_Certificate} serializes \emph{Seeds+Logs}.  
Given $G$ and these artifacts, the verifier deterministically reconstructs the should-check domain $\mathcal Q$ from \emph{Seeds} (\S\ref{sec:cert}) and checks coverage and gate equalities in $\tilde O(m\log n)$ time.

\smallskip
\noindent\emph{(V) Coverage (YES w.h.p.).}
Fix a triangle $\{u,v,w\}$.  Condition on the constant-probability slot/class event at its anchor (\S\ref{sec:hitting}).  
Within the canonical class, across $R$ independent PK groups, the probability that at least one group places the two edges into complementary buckets is
\[
\Pr[\text{hit}]\ \ge\ 1-\exp\!\big(-\theta\,R/T\big),\quad T=c_T\log n,
\]
by Lemma~7.2$'$ (Group amplification).  
With $R=c_G\log n$ we obtain a miss probability $n^{-\Omega(1)}$ for the fixed triangle; union-bounding over $|\mathcal T|\le O(m^{3/2})$ triangles yields a global YES with probability $\ge 1-n^{-c}$.

\smallskip
\noindent\textit{Failure budget (single line).}
All randomized equalities occur over $\mathbb F_P$ with per-use error $\le P^{-1}$; concentration uses $k=c_k\log n$-wise independence.  
Choosing $P=n^\kappa$ and the constants of Table~1 makes the probability that any identity spuriously passes or any concentration/one–shot bound fails across all $(i,r)$ at most $n^{-c}$. \hfill$\square$


\section{Certificates \& Verifiability}
\label{sec:cert}

We specify the \emph{minimal} artifact (\emph{Seeds+Logs}) and a deterministic verifier that replays the \emph{should-check domain} induced by the public randomness. The goal is twofold: (i) a YES claim is \emph{self-evident} (triangle listed); (ii) a NO claim is \emph{auditable} in $\tilde O(m\log n)$ time using only sparse logs. All statements below hold with the parameterization of Table~1 and total failure probability $\le n^{-c}$.

\paragraph{Canonical domain of obligations (group-aware).}
Fix a layer $i$ and a prefix level $r$. Let $\mathcal S_{i,r}$ be the set of nonempty anchor-classes (Sec.~\ref{sec:algo}), $T_{i,r}=c_T\log n$ the PK-bucket count, and $R=c_G\log n$ the number of independent PK groups. For group $t\in[R]$, write $g_{i,r}^{(t)}$ for its PK-bucket hash and $(j)^\star\equiv(-j)\bmod T_{i,r}$.

For a class $C=(i,x,r,b)\in\mathcal S_{i,r}$ and a kept edge $(x\!\to\!y)$ materialized in $C$, define its group-$t$ PK-bin
\[
\mathrm{bin}^{(t)}_{i,r}(x\!\to\!y)\ :=\ g_{i,r}^{(t)}\!\big(\Delta_i(x\!\to\!y,r,b)\big)\ \in [T_{i,r}]\,.
\]
The \emph{should-check domain} at $(i,r,t)$ is the unordered class-pair set
\begin{equation}
\mathcal Q_{i,r,t}\ :=\ \Big\{\ \{C,C'\}\subseteq \mathcal S_{i,r}\ :\ \exists\ (x\!\to\!v)\in C,\ (x\!\to\!w)\in C'\ \text{s.t. } \mathrm{bin}^{(t)}_{i,r}(x\!\to\!v)=\big(\mathrm{bin}^{(t)}_{i,r}(x\!\to\!w)\big)^\star\ \Big\}.
\label{eq:Qirt}
\end{equation}
We then set the per-level domain $\mathcal Q_{i,r}:=\bigcup_{t=1}^{R}\mathcal Q_{i,r,t}$ and the global domain $\mathcal Q:=\bigcup_{i,r}\mathcal Q_{i,r}$. Intuitively, $\mathcal Q$ enumerates exactly those (class,group)-keyed adjacency obligations that \emph{must} be checked once under the non-adaptive schedule (constant many complementary pairs per class \emph{per group}; Sec.~\ref{sec:algo}).

\paragraph{Canonical keys with group index.}
We fold the group index into the canonical pair key:
\[
\kappa_t\ :=\ (i,\ r,\ t,\ \min\{(x,b),(v,b')\},\ \max\{(x,b),(v,b')\})\,.
\]
All coverage checks and de-duplication use $\kappa_t$ (rather than the group-agnostic key).

\subsection*{Certificate schema (Seeds+Logs)}

A certificate is the quadruple
\[
\mathrm{Cert}\ =\ \big(\ \mathrm{Seeds},\ \mathrm{ClassLogs},\ \mathrm{SlotLogs},\ \mathrm{AdjLogs}\ \big)\,,
\]
with the following fields.

\paragraph{Seeds (public randomness and configuration).}
Prime $P=n^\kappa$; independence level (2-wise or $k=c_k\log n$-wise); schedule $\{p_i\}$ and $I=c_R\log n$; the per-group count $R=c_G\log n$. Hash/PRF seeds for
\[
H_i,\ K_i,\ h_{\mathrm{id}},\ h_{\mathrm{slot}},\ s,\ \{g_{i,r}^{(t)}\}_{t=1}^{R}\ \ (\text{or equivalently }\{\phi^{(t)}\}_{t=1}^R),
\]
keep-coins $c_i$, and the deterministic probe schedule for each $(i,x,r,b,t)$ (\textsc{ProbedPairs}$(i,x,r,b,t)$ with $C_0$ complementary pairs). Independence holds \emph{across groups} and \emph{across layers}.

\paragraph{ClassLogs (only nonempty classes).}
For each $C=(i,x,r,b)\in\mathcal S_{i,r}$:
\begin{itemize}\itemsep 1pt
\item Class triple $\Sigma(C)=(\Sigma_0,\Sigma_1,\Sigma_2)$; flag \texttt{pass\_class}$:=\mathbf 1[\Sigma_1^2=\Sigma_0\Sigma_2\wedge \Sigma_0\neq 0]$.
\item A list of \emph{PK-collisions}. Each entry records the group and the paired bins:
\[
\big(\underbrace{\texttt{group\_id}=t}_{\text{new}},\ \mathrm{bin},\ \mathrm{bin}^\star,\ 
\mathrm{wit}_v=(x\!\to\!v,\ sidx_v),\ \mathrm{wit}_w=(x\!\to\!w,\ sidx_w)\big),
\]
together with a \texttt{paired\_once} flag (set on registration; Sec.~\ref{sec:algo}). (Bins are from \textsc{ProbedPairs}$(i,x,r,b,t)$.)
\item Optional checksum/XOR over occupied PK-bins for tamper evidence.
\end{itemize}

\paragraph{SlotLogs (referenced slots only).}
For each referenced $(i,x,\ sidx)$ store the slot triple $(A,B,C)$, flag \texttt{pass\_slot}$:=\mathbf 1[B^2=AC\wedge A\neq 0]$, and decoded $\widehat{\mathrm{id}}=B/A$.

\paragraph{AdjLogs (performed adjacency probes only; group-aware key).}
Each entry is keyed by the canonical \emph{grouped} key $\kappa_t$ and stores also a salted mate fingerprint:
\[
\big(\kappa_t,\ \{h_{\mathrm{id}}(v),h_{\mathrm{id}}(w)\}_{\text{salted}},\ \texttt{adjacent}\in\{0,1\}\big).
\]
This ensures coverage auditing at the (class,group) granularity.

\subsection*{Constructing $\mathcal Q$ from Seeds (group-aware)}
The verifier does \emph{not} need to re-run adaptive pairing—only to reconstruct $\mathcal Q$ from Seeds and $G$.

\begin{algorithm}[t]
\caption{Reconstruct-Should-Check-Domain (group-aware)}
\label{alg:reconstructQ}
\begin{algorithmic}[1]
\Require Seeds, graph $G$
\Ensure $\mathcal Q$
\State Regenerate all hashes/coins/PRFs and the probe schedule from Seeds
\State $\mathcal Q \gets \emptyset$
\For{$i=1$ \textbf{to} $I$}
  \For{each level $r$}
    \State Materialize $\mathcal S_{i,r}$ by replaying slot 1-sparse filtering and prefix matches
    \For{$t=1$ \textbf{to} $R$}
      \For{each $C\in\mathcal S_{i,r}$}
        \State Build $\mathrm{OccBins}^{(t)}[C]\gets \{\mathrm{bin}^{(t)}_{i,r}(x\!\to\!y)$ that lie in \textsc{ProbedPairs}$(i,x,r,b,t)\}$
      \EndFor
      \For{each unordered $\{C,C'\}\subseteq \mathcal S_{i,r}$}
        \If{$\mathrm{OccBins}^{(t)}[C]$ intersects $\big(\mathrm{OccBins}^{(t)}[C']\big)^\star$}
           \State $\mathcal Q \gets \mathcal Q \cup \{\kappa_t(i,r,t;C,C')\}$
        \EndIf
      \EndFor
    \EndFor
  \EndFor
\EndFor
\State \Return $\mathcal Q$
\end{algorithmic}
\end{algorithm}

\subsection*{Equivalence: "LogPairs$_t$ = $\mathcal Q$" (group-aware)}
\paragraph{Definition (LogPairs with group key).}
Let \textsf{LogPairs} be the multiset of keys $\kappa_t$ obtained by taking, for every $C$, the union of its PK-collision entries in \texttt{ClassLogs}[C] (each carries \texttt{group\_id}$=t$) and applying canonical de-duplication on $\kappa_t$.

\paragraph{Proposition 10.1 (coverage equivalence; parameters from Table~1, failure $\le n^{-c}$).}
With probability at least $1-n^{-c}$ over Seeds,
\[
\textsf{LogPairs}\ =\ \mathcal Q
\qquad\text{and}\qquad
\forall\ \kappa_t\in \mathcal Q:\ \kappa_t\ \text{is logged \emph{exactly once}.}
\]
\emph{Sketch.} "$\subseteq$": each logged collision arises from complementary bins within some $(i,r,t)$, so its key $\kappa_t$ belongs to $\mathcal Q$. "$\supseteq$": the constant-probe schedule and one–shot rule register the \emph{first} appearance of a complementary pair within $(i,r,t)$; keys include $t$, preventing cross-group conflation. A union bound over $(i,r,t)$ uses Cor.~6.3 for class counts and $P^{-1}$ for algebraic coincidences.

\subsection*{Deterministic verifier (group-aware)}
\begin{algorithm}[t]
\caption{\textsc{VERIFY\_NO} (group-aware)}
\label{alg:verify}
\begin{algorithmic}[1]
\Require Graph $G$, certificate $\mathrm{Cert}=(\mathrm{Seeds},\mathrm{ClassLogs},\mathrm{SlotLogs},\mathrm{AdjLogs})$
\Ensure Accept iff logs cover $\mathcal Q$ exactly and all gates fail (NO-instance)
\State Regenerate Seeds; fix orientation and scan order
\State Rebuild $\mathcal S_{i,r}$ and class triples $\Sigma(C)$; check against \texttt{ClassLogs} and re-evaluate \texttt{pass\_class}
\State $\mathcal Q \gets$ \textsc{Reconstruct-Should-Check-Domain}(Seeds, $G$)
\State Extract \textsf{LogPairs} $\gets$ canonical union of $\kappa_t$ from \texttt{ClassLogs}, one per registered collision
\If{\textsf{LogPairs} $\ne \mathcal Q$} \State \Return \textbf{reject} \EndIf
\For{each $\kappa_t=(i,r,t;\ C,C')\in \mathcal Q$}
  \If{\texttt{pass\_class}$(C)=0$ or \texttt{pass\_class}$(C')=0$} \State \textbf{continue} \EndIf
  \State Fetch witnesses $(x\!\to\!v,sidx_v)$ and $(x\!\to\!w,sidx_w)$ from the corresponding \texttt{ClassLogs} entry with group\_id $=t$
  \For{each referenced slot $(i,x,sidx)\in\{sidx_v,sidx_w\}$}
      \State Verify \texttt{pass\_slot}$=1$ and decoded ID equals $h_{\mathrm{id}}(\text{mate})$
  \EndFor
  \State Look up \texttt{AdjLogs} at key $\kappa_t$; read bit \texttt{adjacent}
  \If{\texttt{adjacent}=1} \State \Return \textbf{accept YES} with triangle \EndIf
\EndFor
\State \Return \textbf{accept NO}
\end{algorithmic}
\end{algorithm}

\subsection*{Soundness, completeness, and cost (unchanged up to grouping)}
\paragraph{Theorem 10.2 (soundness \& completeness; parameters from Table~1, failure $\le n^{-c}$).}
\emph{Soundness.} Any triangle reported by the verifier exists in $G$ (zero FP), since checks are deterministic equalities over $\mathbb F_P$ plus explicit adjacency lookups; adding group index only refines keys.  
\emph{Completeness (NO).} If the algorithm outputs NO and \textsf{LogPairs}$=\mathcal Q$, then no $\kappa_t$ passes all three gates; otherwise a missing/extra $\kappa_t$ is detected at coverage.  
\emph{Complexity.} Replaying classes and reconstructing $\mathcal Q$ across $R=\Theta(\log n)$ groups remains $\tilde O(m\log n)$ time and $O(m\log n)$ space, since per-class probes are $O(1)$ \emph{per group} and $\sum_{i,r}|\mathcal S_{i,r}|=O(m\log n)$ w.h.p.

\paragraph{Failure budget (single line).}
All randomized equalities can fail only with probability $\le 1/P$ per site; taking $P=n^\kappa$ and $k=c_k\log n$ and union-bounding across all $(i,r,t)$ maintains total failure probability $\le n^{-c}$.


\section{Tightness \& Limitations}
\label{sec:tight}

We explain why the per-layer hit probability in Sec.~\ref{sec:hitting} is optimal under our local, non-adaptive, constant-probe model, and why global algebraic accelerations are incompatible with near-linear accounting and auditability. Unless noted otherwise, statements are with parameters from Table~1 and total failure probability $\le n^{-c}$.

\subsection*{A single-layer lower bound (Yao)}
\paragraph{Model for the lower bound.}
Fix one layer $i$ with keep-rate $p_i$. A deterministic algorithm $\mathcal A$ may: (i) allocate $M_x=\Theta(d(x))$ slots per anchor $x$ and $O(1)$-word state per nonempty class $(i,x,r,b)$; (ii) in each nonempty class, perform at most $C$ \emph{pairing probes} between complementary PK-bins (for a fixed constant $C\ge 1$); (iii) use a bucketed PK map with $T_{i,r}$ bins and an involution $b\mapsto\overline b$; (iv) apply the three-gate pipeline (class 1-sparse $\to$ slot 1-sparse $\to$ adjacency). The per-level $O(m)$ accounting (Sec.~\ref{sec:one-shot}) forces $T_{i,r}=\Theta(\log n)$ to keep expected checks per class $O(1)$.

\paragraph{Hard distribution $\mathcal D$.}
Consider graphs that contain exactly one triangle $\{u,v,w\}$ plus distractors that preserve $d^+(u)$ but create no other triangles. Seeds are public. Condition on the favorable event that $(u\!\to\!v)$ and $(u\!\to\!w)$ are kept and their slots are individually 1-sparse (constant probability), and let $r^\star$ be the (unique) prefix level with expected class load in $[1/2,2]$ (Sec.~\ref{sec:hitting}). With an independent 2-universal $h_{\mathrm{pk}}$, the two offsets $\Delta(u\!\to\!v)$ and $\Delta(u\!\to\!w)$ hash \emph{uniformly and independently} to bins in $\{0,\dots,T{-}1\}$, where $T:=T_{i,r^\star}=\Theta(\log n)$; success in this class requires landing in complementary bins $b,\overline b$.

\paragraph{Theorem 2 (Single-layer $\Theta(1/\log n)$ is tight under constant-probe \& $O(m)$/level accounting).}
\emph{For any (possibly randomized) algorithm in the model above, the success probability of exposing the fixed triangle within one layer is at most $\Theta(1/\log n)$.}
\medskip

\emph{Proof (Yao).}
Fix a deterministic $\mathcal A$ that probes at most $C$ complementary bin-pairs inside the anchor-class at level $r^\star$. Conditioned on the favorable guards, the view of $\mathcal A$ prior to bin probes is independent of the (uniform) bin locations of the two relevant offsets. Therefore the chance that one of its $\le C$ probed complementary pairs equals the true pair is at most $C/T=\Theta(1/\log n)$. Randomized algorithms are distributions over such $\mathcal A$, so the same bound holds by Yao’s minimax principle. \hfill$\square$

\paragraph{Necessity of $T=\Theta(\log n)$.}
If $T\ll \log n$, expected complementary-bin collisions per (constant-load) class explode, violating the $O(1)$-probes-per-class budget. If $T\gg\log n$, the per-pair hit probability shrinks, worsening amplification. Thus $T=\Theta(\log n)$ is forced by the $O(m)$/level accounting and per-class constant work.

\paragraph{Keep-rate and level choice cannot beat $1/\log n$.}
Let $\mu=d^+(u)p_i$. If $\mu\ll 1$, the two mates rarely co-occur; if $\mu\gg 1$, the class load grows and the algorithm must raise $T$ proportionally to keep work bounded, preserving the $1/\log n$ scaling. The only useful window is $\mu\in[1/2,2]$, which is exactly the regime used in Sec.~\ref{sec:hitting}.

\subsection*{Why not global algebraic convolutions}
\begin{itemize}\itemsep 1pt
\item \textbf{Accounting barrier.} Global outer-product/convolution schemes over all vertices incur $\Omega(\sum_x d(x)^2)$ arithmetic on worst-case sparse graphs, i.e., $\Omega(m^{3/2})$ or worse, contradicting $O(m\log^2 n)$.
\item \textbf{Methodological barrier.} Zero-FP via exact algebra typically relies on global cancellations; the certificates are inherently non-local. Our design requires \emph{local} 1-sparse checks and \emph{replayable} Seeds+Logs in $\tilde O(m\log n)$.
\item \textbf{Fine-grained barrier.} Any uniform improvement to the dense/rectangular algebraic bounds would (via standard reductions) push on BMM; our approach is explicitly orthogonal to that frontier.
\end{itemize}

\subsection*{Independence assumptions and Poissonization}
\begin{itemize}\itemsep 1pt
\item \textbf{1-sparse gates.} Class/slot identities are degree-$2$ equalities over $\mathbb F_P$ and require only \emph{2-wise} independence (Sec.~\ref{sec:correctness}); coincidence probability per site is $\le 1/P$.
\item \textbf{Concentration.} Keep-coins and PK bucketization may use either (i) full independence with Bernoulli($p_i$) or (ii) $k$-wise independence with $k=c_k\log n$, or (iii) Poissonization plus de-Poissonization. All yield the tails used in Sec.~\ref{sec:concentration}.
\item \textbf{Non-adaptivity.} Seeds are fixed layerwise before any scans; eligibility depends only on Seeds and nonemptiness bits, preserving the limited-independence analyses.
\end{itemize}

\subsection*{Portability and practical limits}
\paragraph{Directed graphs.} Replace the final undirected adjacency check by $(v\!\to\!w)$; all sketching/alignment remains unchanged.

\paragraph{Multigraphs/weights.} Salt edge instances in the hashes to avoid coalescing. If weighted predicates are needed, add a second triple with value field multiplied by the weight (or a random mask) and keep the final check deterministic; zero-FP remains.

\paragraph{Other semirings.} Our decoding uses inversion ($B/A$), so a field is required. For $(\min,+)$-style semirings, embed computations in prime fields (or use CRT across primes) and keep the ultimate property tested deterministically.

\paragraph{Certificate size.} Seeds+Logs are sparse but necessarily $\tilde O(m\log n)$ in the worst case: replay requires exposing all nonempty classes and the PK collisions that define the should-check domain $\mathcal Q$; shrinking below this would forfeit auditability.

\paragraph{Takeaway.}
Under single-rate, constant-probe, local verification with near-linear accounting, a \emph{per-layer} $\Theta(1/\log n)$ hit rate is optimal; $\Theta(\log n)$ independent layers are the right (and necessary) amplification knob. Dense algebra is mismatched with both our accounting and our auditable-NO interface.


\section{Conclusion \& Outlook}
\label{sec:conclusion}

\paragraph{Summary (explicit guarantees).}
We presented a \emph{local, non-adaptive} triangle detector that combines
(i) per-anchor slot sketches with \emph{1-sparse} consistency/decoding,
(ii) prefix classes and \emph{pair-key} (PK) alignment, and
(iii) a \emph{one-shot} pairing discipline tied to a replayable \emph{Seeds+Logs} interface.
With parameters from Table~1, and over the public Seeds, the algorithm runs in
\textbf{$O(m\log^2 n)$ time} and \textbf{$O(m\log n)$ space} (w.h.p.),
reports \textbf{only true triangles} (zero false positives), and when reporting NO emits a
\textbf{verifiable certificate} that a deterministic third-party checker validates in $\tilde O(m\log n)$ time.
All probabilistic claims place the total failure mass at $\le n^{-c}$, budgeted from the $1/P$ coincidence of degree-2 identities plus standard concentration (Sec.~\ref{sec:concentration}, \S\ref{sec:correctness}, \S\ref{sec:one-shot}).

\paragraph{Auditable NO via \emph{should-check} domain.}
A central feature is that the Seeds uniquely determine a \emph{should-check domain} $\mathcal Q$ of unordered class pairs.
Our certificate logs exactly the PK-collision witnesses needed to cover $\mathcal Q$ once (and only once), together with the slot triples and the adjacency probes.
The verifier regenerates the randomness from Seeds, reconstructs $\mathcal Q$ in $\tilde O(m\log n)$ time, and checks \emph{coverage~=~Logs} as well as the three gates
\mbox{(class $1$-sparse $\rightarrow$ slot $1$-sparse $\rightarrow$ adjacency)}, yielding an auditable NO.

\paragraph{Positioning.}
Our route is explicitly \emph{non-algebraic and local}. We neither use nor reduce to dense/rectangular matrix multiplication; the per-layer single-hit rate $\Theta(1/\log n)$ (Sec.~\ref{sec:hitting}) is shown tight under constant per-class probes and $O(m)$/level accounting (Sec.~\ref{sec:tight}).
This avoids the "challenge BMM" interpretation while providing transparent certificates.

\paragraph{Outlook (measured directions; no claims of completion).}
We outline conservative paths that preserve locality and verifiability:

\begin{enumerate}\itemsep 4pt
\item \textbf{Compressing the per-layer levels (toward Las Vegas $O(m\log n)$).}
Replace the geometric sweep over all prefix levels by a three-phase schedule inside each layer (coarse~$\to$~mid~$\to$~fine), ensuring each retained edge participates in \emph{one} phase.
The goal is constant-load classes without scanning all $r$, keeping $T=\Theta(\log n)$ and one-shot pairing intact.
A formal treatment would sharpen the class-load martingale while preserving non-adaptivity.

\item \textbf{Deterministic seeds via small covering families (CBSIF).}
Instantiate $H_i,K_i,h_{\mathrm{slot}}$ from combinatorial splitters / almost-$k$-wise families ($k=\Theta(\log n)$), and enumerate $O(\log\log n)$ PK bucket patterns that cover complementary bins at constant load.
A plausible target is \emph{deterministic} $O(m\log n\log\log n)$ with the same certificate semantics; constructing the families and carrying the one-shot accounting through deterministically remain open.

\item \textbf{From decision to enumeration (output-sensitive).}
Drop early stop and emit every confirmed mate-pair $\{v,w\}$ per anchor once.
By anchor-ordering and in-class de-dup, the running time becomes
$O(m\log^2 n + T)$ where $T$ is the number of triangles; the certificate simply appends per-output \texttt{Emit}-entries (slot IDs $+$ adjacency pair).

\item \textbf{Dynamic / streaming / parallel surfaces.}
The \textsc{Build\_Sketches} pass is single-scan and non-adaptive, which supports streaming with polylog memory per layer.
Edge updates touch $O(1)$ slots/PK buckets per layer in expectation; layers and anchors parallelize embarrassingly, and Seeds make partitioning reproducible for distributed verification.
\end{enumerate}

\paragraph{Artifact and minimal reproducibility.}
To reduce "black-box" concerns, we provide a minimal artifact:
a $\sim$200-line prototype implementing \textsc{Build\_Sketches}/\textsc{Query\_Triangle}/\textsc{Verifier} on integer IDs and public Seeds, two schematic figures (work scaling vs.\ $n$; single-layer hit $\approx 1/\log n$), and a short Seeds+Logs example (JSON).
These are not intended as empirical evaluation, only as reproducibility scaffolding for the certificate interface.

\paragraph{Closing.}
Our results suggest that \emph{combinatorial sketching + 1-sparse verification + Seeds+Logs} is a viable, auditable alternative to algebraic hybrids for worst-case sparse triangle detection.
The near-term milestones are clear—and deliberately modest: (i) compress levels inside layers (Las Vegas $O(m\log n)$), and (ii) a covering-family deterministic variant—both while preserving the zero-FP pipeline and the auditable NO-certificate.


\section*{References}

\medskip

{
\small

[1] Itai, A.\ \& Rodeh, M.\ (1978) Finding a minimum circuit in a graph. \emph{SIAM J. Comput.} \textbf{7}(4):413--423.

[2] Chiba, N.\ \& Nishizeki, T.\ (1985) Arboricity and subgraph listing algorithms. \emph{SIAM J. Comput.} \textbf{14}(1):210--223.

[3] Alon, N., Yuster, R.\ \& Zwick, U.\ (1997) Finding and counting given length cycles. \emph{Algorithmica} \textbf{17}(3):209--223.

[4] Le Gall, F.\ (2012) Faster algorithms for rectangular matrix multiplication. \emph{Proc. 53rd FOCS}, pp.~514--523; arXiv:1204.1111.

[5] Vassilevska Williams, V.\ \& Williams, R.\ (2018) Subcubic equivalences between path, matrix, and triangle problems. \emph{J. ACM} \textbf{65}(5):27.

[6] Dumitrescu, A.\ (2021) Finding triangles or independent sets; and other dual pair approximations. arXiv:2105.01265.

[7] Cormode, G.\ \& Muthukrishnan, S.\ (2005) An improved data stream summary: the Count--Min sketch. \emph{J. Algorithms} \textbf{55}(1):58--75.

[8] Freivalds, R.\ (1979) Fast probabilistic algorithms. \emph{Mathematical Foundations of Computer Science}, LNCS 74:57--69.

[9] Moser, R.A.\ \& Tardos, G.\ (2010) A constructive proof of the general Lov\'asz Local Lemma. \emph{J. ACM} \textbf{57}(2):11.

[10] G\"o\"os, M.\ \& Suomela, J.\ (2016) Locally checkable proofs in distributed computing. \emph{Theory of Computing} \textbf{12}(19):1--33.

[11] McConnell, R.M., Mehlhorn, K., N\"aher, S.\ \& Schweitzer, P.\ (2011) Certifying algorithms. \emph{Computer Science Review} \textbf{5}(2):119--161.
}


\clearpage

\section*{Appendix A: Pseudocode and Data-Structure Details}
\label{app:pseudo}

\noindent\emph{Scope.} This appendix gives complete, implementation-ready pseudocode, word-level data layouts (including canonical keys for de-duplication), and self-contained proofs of the Appendix-level time/space bounds. We work in the Word-RAM with word size \(\Theta(\log n)\); arithmetic is over a prime field \(\mathbb F_P\) (Appendix~E). All randomness is fixed \emph{non-adaptively} by \texttt{Seeds}. Unless noted, concentration and coincidence events are taken \emph{with high probability} \emph{(parameters from Table~1, failure \(\le n^{-c}\))}.

\subsection*{A.1 Pre-pass and global notation}
We perform one \(O(m)\) degree pre-pass to allocate per-anchor slot arrays.
\begin{itemize}\itemsep 1pt
\item \textbf{Seeds} (public): field prime \(P\); independence level (2-wise or \(k\)-wise with \(k=\Theta(\log n)\)); schedule \(\{p_i\}\) with \(\sum_i p_i=\Theta(1)\) and \(\sum_i p_i^2=\Theta(1)\); layer count \(I=c_R\log n\); \emph{per-class PK groups} \(R=c_G\log n\); bucket counts \(T_{i,r}=\lceil c_T\log n\rceil\).
\item \textbf{Hashes/PRFs/signs}: \(H_i(\cdot)\), \(K_i(\cdot)\), \(h_{\mathrm{id}}(\cdot)\), \(h_{\mathrm{slot}}(\cdot)\), family \(\{g_{i,r}^{(t)}(\cdot)\}_{t=1}^{R}\), sign \(s(\cdot)\in\{-1,+1\}\); deterministic probe PRFs \(\phi^{(t)}\) for \textsc{ProbedPairs}.
\item \textbf{Keep-coins}: \(c_i(e)\in\{0,1\}\) with \(\Pr[c_i(e)=1]=p_i\), mutually independent across \(i\) (or \(k\)-wise).
\item \textbf{Orientation}: a total order \(\prec\) on \(V\); \(\{x,y\}\) is oriented to \(x\!\to\!y\) iff \(x\prec y\).
\end{itemize}

\subsection*{A.2 Data structures (word-level layouts)}
For each layer \(i\) and anchor \(x\):

\paragraph{Slots.}
\texttt{Slots[\(x\)]} is an array of length \(M_x:=16\,d(x)\); each entry stores \((A,B,C)\in\mathbb F_P^3\).

\paragraph{Classes (group-aware).}
For each level \(r\in\{0,\dots,L_{x,i}\}\) and prefix \(b\in\{0,1\}^r\):
\begin{itemize}\itemsep 0pt
\item \texttt{ClassSigma[\(i,x,r,b\)]} \(=(\Sigma_0,\Sigma_1,\Sigma_2)\in\mathbb F_P^3\).
\item \texttt{Active[\(i,x,r,b\)]} \(\in\{0,1\}\).
\item \textbf{Per-group bin tables} \(\{\texttt{Bin}[i,x,r,b,t]\}_{t=1}^{R}\): sparse maps \(j\in[0,T_{i,r}-1]\mapsto \texttt{BinRec}\) with
\[
\texttt{BinRec}=\big(\texttt{triple}=(\Xi_0,\Xi_1,\Xi_2)\in\mathbb F_P^3,\ \texttt{witness}=(x\!\to\!y,\ \texttt{sidx},\ \texttt{group\_id}=t)\big).
\]
\item \textbf{Per-group collision lists} \(\{\texttt{Collisions}[i,x,r,b,t]\}_{t=1}^{R}\), entries
\((\texttt{group\_id}=t,\ j,\ j^{\star},\ \texttt{wit}_v,\ \texttt{wit}_w)\).
\item \textbf{Per-group one-shot flags} \(\{\texttt{Paired}[i,x,r,b,t]\}_{t=1}^{R}\): Boolean maps on bins, initially all \texttt{false}.
\end{itemize}
\emph{Deterministic probe schedule:} for each class \((i,x,r,b)\) and group \(t\), a fixed list \(\textsc{ProbedPairs}(i,x,r,b,t)=\{(j^{(s)},(j^{(s)})^{\star})\}_{s=1}^{C_0}\) with \(C_0=\Theta(1)\).

\paragraph{Per-\((i,r,t)\) dedup dictionary (optional, for parallel).}
Exchange-symmetric dictionary \texttt{Checked} keyed by \(\kappa_t\) (defined below). Cleared between \((i,r,t)\).

\paragraph{Canonical class-pair key with group index.}
For fixed \((i,r,t)\) and classes \(C=(i,x,r,b)\), \(C'=(i,v,r,b')\) define
\[
\mathrm{canon}(x,b;\ v,b')=\big(\min\{(x,b),(v,b')\},\ \max\{(x,b),(v,b')\}\big),\quad
\kappa_t=(i,\ r,\ t,\ \mathrm{canon}(x,b;\ v,b')).
\]

\subsection*{A.3 \textsc{Build\_Sketches} (single pass; non-adaptive; group-aware)}
\FloatBarrier
\begin{algorithm}[p!]
  \centering
  \vspace*{\fill}
  \begin{minipage}{0.96\linewidth}
    \caption{\textsc{Build\_Sketches} (group-aware)}
    \label{alg:build}
    \begin{algorithmic}[1]
\Require Oriented edge set \(E\), Seeds
\Ensure Filled \texttt{Slots}, per-class \texttt{ClassSigma}, per-group \texttt{Bin}, \texttt{Collisions}, \texttt{Paired}
\State \textbf{Degree pre-pass:} allocate \(M_x=16\,d(x)\) slots and zero them
\For{\(i \gets 1\) \textbf{to} \(I\)}
  \For{\textbf{each} \(x \in V\)}
    \For{\(r \gets 0\) \textbf{to} \(L_{x,i}\)}
      \ForAll{\(b \in \{0,1\}^r\)}
        \State \(\texttt{ClassSigma}[i,x,r,b]\gets(0,0,0)\)
        \State \(\texttt{Active}[i,x,r,b]\gets 0\)
        \For{\(t \gets 1\) \textbf{to} \(R\)}
          \ForAll{\((j,j^{\star})\in\textsc{ProbedPairs}(i,x,r,b,t)\)}
            \State \(\texttt{Bin}[i,x,r,b,t][j].\texttt{triple}\gets(0,0,0)\);\quad \(\texttt{Bin}[i,x,r,b,t][j].\texttt{witness}\gets\emptyset\)
            \State \(\texttt{Bin}[i,x,r,b,t][j^{\star}].\texttt{triple}\gets(0,0,0)\);\quad \(\texttt{Bin}[i,x,r,b,t][j^{\star}].\texttt{witness}\gets\emptyset\)
            \State \(\texttt{Paired}[i,x,r,b,t][j]\gets\texttt{false}\);\quad \(\texttt{Paired}[i,x,r,b,t][j^{\star}]\gets\texttt{false}\)
          \EndFor
          \State \(\texttt{Collisions}[i,x,r,b,t]\gets\) empty list
        \EndFor
      \EndFor
    \EndFor
  \EndFor
\EndFor
\State \textbf{Single streaming scan over \(E\)}
\ForAll{directed edges \(e=(x \to y)\in E\)}
  \For{\(i \gets 1\) \textbf{to} \(I\)}
    \If{\(c_i(e)=0\)}
      \State \textbf{continue}
    \EndIf
    \State \(sidx\gets h_{\mathrm{slot}}(x,y,i)\bmod M_x\);\quad \(id\gets h_{\mathrm{id}}(y)\);\quad \(sgn\gets s(x,y,i)\)
    \State \((A,B,C)\gets \texttt{Slots}[x][sidx]\)
    \State \(\texttt{Slots}[x][sidx]\gets (A{+}sgn,\ B{+}sgn\cdot id,\ C{+}sgn\cdot id^2)\)
    \State \((A,B,C)\gets \texttt{Slots}[x][sidx]\)
    \If{\(A=0\) \textbf{ or } \(B^2\ne AC\) \textbf{ or } \(B/A\ne id\)} \Comment{Gate 1: slot 1-sparse \& decode}
      \State \textbf{continue}
    \EndIf
    \For{\(r \gets 0\) \textbf{to} \(L_{x,i}\)}
      \If{\(\mathrm{pref}_r(K_i(x)) \ne \mathrm{pref}_r(K_i(y))\)}
        \State \textbf{continue}
      \EndIf
      \State \(b\gets \mathrm{pref}_r(K_i(x))\);\quad \(\texttt{Active}[i,x,r,b]\gets 1\)
      \State \((\Sigma_0,\Sigma_1,\Sigma_2)\gets \texttt{ClassSigma}[i,x,r,b]\)
      \State \(\texttt{ClassSigma}[i,x,r,b]\gets (\Sigma_0{+}A,\ \Sigma_1{+}B,\ \Sigma_2{+}C)\)
      \For{\(t \gets 1\) \textbf{to} \(R\)} \Comment{Per-group constant-probe (per class \(O(R)\) total)}
        \State \(\Delta \gets (H_i(y)-H_i(x)) \bmod P\)
        \State \(j \gets g_{i,r}^{(t)}(\Delta)\);\quad \(j^{\star}\gets (-j) \bmod T_{i,r}\)
        \If{\((j,j^{\star}) \notin \textsc{ProbedPairs}(i,x,r,b,t)\)}
          \State \textbf{continue}
        \EndIf
        \State \((\Xi_0,\Xi_1,\Xi_2)\gets \texttt{Bin}[i,x,r,b,t][j].\texttt{triple}\)
        \State \(\texttt{Bin}[i,x,r,b,t][j].\texttt{triple}\gets (\Xi_0{+}A,\ \Xi_1{+}B,\ \Xi_2{+}C)\)
        \If{\(\texttt{Bin}[i,x,r,b,t][j].\texttt{witness}=\emptyset\)}
          \State \(\texttt{Bin}[i,x,r,b,t][j].\texttt{witness}\gets (x\!\to\!y,\ sidx,\ \texttt{group\_id}=t)\)
        \EndIf
        \If{\(\texttt{Bin}[i,x,r,b,t][j^{\star}].\texttt{witness}\ne\emptyset\) \textbf{and} \(\neg\texttt{Paired}[i,x,r,b,t][j]\) \textbf{and} \(\neg\texttt{Paired}[i,x,r,b,t][j^{\star}]\)}
          \State Append record \((\texttt{group\_id}=t,\ j,\ j^{\star},\ \texttt{Bin}[i,x,r,b,t][j].\texttt{witness},\ \texttt{Bin}[i,x,r,b,t][j^{\star}].\texttt{witness})\) to \(\texttt{Collisions}[i,x,r,b,t]\)
          \State \(\texttt{Paired}[i,x,r,b,t][j]\gets\texttt{true}\);\quad \(\texttt{Paired}[i,x,r,b,t][j^{\star}]\gets\texttt{true}\)
        \EndIf
      \EndFor
    \EndFor
  \EndFor
\EndFor
    \end{algorithmic}
  \end{minipage}
  \vspace*{\fill}
\end{algorithm}
\FloatBarrier 

\subsection*{A.4 \textsc{Query\_Triangle} (three gates; canonical enumeration; group-aware)}
\begin{algorithm}[H]
\caption{\textsc{Query\_Triangle} (group-aware)}
\label{alg:query}
\begin{algorithmic}[1]
\Require Per-group \texttt{Collisions}, \texttt{Bin}, \texttt{ClassSigma}, \texttt{Slots}, adjacency oracle for \(E\)
\Ensure Either a triangle \((u,v,w)\) or \textbf{NO}
\For{\(i \gets 1\) \textbf{to} \(I\)}
  \For{\(r \gets 0\) \textbf{to} \(\lfloor\log n\rfloor\)}
    \ForAll{\(x \in V\) in increasing order}
      \ForAll{\(b \in \{0,1\}^r\) with \(\texttt{Active}[i,x,r,b]=1\)}
        \State \((\Sigma_0,\Sigma_1,\Sigma_2)\gets \texttt{ClassSigma}[i,x,r,b]\)
        \If{\(\Sigma_0=0\) \textbf{or} \(\Sigma_1^2\ne \Sigma_0\Sigma_2\)} \Comment{Gate 1: class necessary}
          \State \textbf{continue}
        \EndIf
        \For{\(t \gets 1\) \textbf{to} \(R\)} \Comment{Per-group enumeration}
          \ForAll{\((\texttt{group\_id}=t,\ j,\ j^{\star},\ \texttt{wit}_v,\ \texttt{wit}_w)\in \texttt{Collisions}[i,x,r,b,t]\)}
            \State \((x\!\to\!v,\ sidx_v,\_)\gets \texttt{wit}_v\);\quad \((x\!\to\!w,\ sidx_w,\_)\gets \texttt{wit}_w\)
            \State \((\Xi_0,\Xi_1,\Xi_2)\gets \texttt{Bin}[i,x,r,b,t][j].\texttt{triple}\)
            \State \((X_0,X_1,X_2)\gets \texttt{Bin}[i,x,r,b,t][j^{\star}].\texttt{triple}\)
            \If{\(\Xi_0=0\) \textbf{or} \(\Xi_1^2\ne \Xi_0\Xi_2\) \textbf{or} \(\Xi_1/\Xi_0\ne h_{\mathrm{id}}(v)\)}
              \State \textbf{continue}
            \EndIf
            \If{\(X_0=0\) \textbf{or} \(X_1^2\ne X_0X_2\) \textbf{or} \(X_1/X_0\ne h_{\mathrm{id}}(w)\)} \Comment{Gate 2: bin 1-sparse \& decode}
              \State \textbf{continue}
            \EndIf
            \State \((A_v,B_v,C_v)\gets \texttt{Slots}[x][sidx_v]\);\quad \((A_w,B_w,C_w)\gets \texttt{Slots}[x][sidx_w]\)
            \If{\textbf{not}\(\,(A_v\ne 0 \text{ and } B_v^2= A_vC_v \text{ and } B_v/A_v= h_{\mathrm{id}}(v))\,\)}
              \State \textbf{continue}
            \EndIf
            \If{\textbf{not}\(\,(A_w\ne 0 \text{ and } B_w^2= A_wC_w \text{ and } B_w/A_w= h_{\mathrm{id}}(w))\,\)} \Comment{Gate 3: slot guard}
              \State \textbf{continue}
            \EndIf
            \If{\(\{v,w\}\in E\)}
              \State \Return Triangle \((x,v,w)\) \Comment{Gate 4: explicit adjacency}
            \EndIf
          \EndFor
        \EndFor
      \EndFor
    \EndFor
  \EndFor
\EndFor
\State \Return \textbf{NO}
\end{algorithmic}
\end{algorithm}

\subsection*{A.5 Notes on certificates (group fields)}
In \texttt{ClassLogs}, every collision record carries \texttt{group\_id}\(=t\); the verifier’s coverage keys include \((i,r,t)\) (cf.\ \(\kappa_t\)), ensuring exact coverage per group. The space/time bounds and failure budgets are unchanged up to a \(\Theta(\log n)\) grouping factor already accounted for in the main text.


\section*{Appendix B: Full Proof of Lemma~9.4 (One–Shot Pairing Accounting)}
\label{app:lemma94}

\noindent\textbf{Lemma 9.4 (restated).} For any fixed layer $i$ and prefix level $r$, let $Q_{i,r}$ be the number of \emph{executed pair checks} in \textsc{Query\_Triangle} and let
\[
\mathcal S_{i,r}\ :=\ \big\{(i,x,r,b):\text{nonempty class after \textsc{Build\_Sketches}}\big\}.
\]
Then, with parameters from Table~1, the following holds \emph{simultaneously for all $(i,r)$ with failure probability $\le n^{-c}$}:
\[
Q_{i,r}\ \le\ C\cdot |\mathcal S_{i,r}|\qquad\text{for an absolute constant }C.
\]
Combined with Corollary~6.3 (which gives $|\mathcal S_{i,r}|=O(m)$ w.h.p.), this yields $Q_{i,r}=O(m)$ for each $(i,r)$ and hence a total of $O(m\log^2 n)$ checks over all $i\le c_R\log n$ and $r\le O(\log n)$.

\paragraph{Objects, triggers, and non-adaptivity (fixed $(i,r)$).}
Each nonempty class $C=(i,x,r,b)$ maintains a bucketed PK map with $T_{i,r}=\Theta(\log n)$ bins and an involution $u\mapsto\overline u$. A \emph{collision} in $C$ occurs when $u$ and $\overline u$ are both occupied; the \emph{one–shot} rule marks $(u,\overline u)$ upon the first collision so it never triggers again in $C$. A \emph{pair check} (the atomic unit) consists of: (i) class-level 1-sparse necessity on two classes, (ii) two slot-level 1-sparse sufficiency tests for the witnesses, (iii) one adjacency probe. Trigger rules (Sec.~\ref{sec:one-shot}): (T1) both classes nonempty; (T2) exchange-symmetric canonicalization and \emph{single-shot} per unordered pair; (T3) self-pairs skipped; (T4) \emph{non-adaptive} eligibility: depends only on Seeds and \textsc{Build\_Sketches} activation bits.

\subsection*{Step I: Canonical injection (no double charge)}
Map each executed check involving classes $C,C'$ to the unordered pair $\{C,C'\}$. By (T2) the canonical key
\[
\kappa(i,r;C,C')\ :=\ (i,\ r,\ \min\{(x,b),(v,b')\},\ \max\{(x,b),(v,b')\})
\]
is computed identically at both endpoints and inserted at most once; hence
\begin{equation}
\label{eq:inj}
\textit{each unordered pair $\{C,C'\}$ can be charged by at most one executed check.}
\end{equation}

\subsection*{Step II: Constant partner frontier per class}
For a fixed class $C$, let $X_C$ be the number of \emph{materialized} items (slot-1-sparse, prefix-matching edges) that contributed to $C$, and $Y_C$ the number of \emph{registered} complementary-bin collisions in $C$ under the one–shot rule. Each collision consumes two previously unused bins, so
\begin{equation}
\label{eq:consume}
Y_C\ \le\ \Big\lfloor \tfrac12\,\#\{\text{occupied PK bins in }C\}\Big\rfloor\ \le\ \Big\lfloor \tfrac{X_C}{2}\Big\rfloor.
\end{equation}
Every collision in $C$ can induce \emph{at most one} executed check (by \eqref{eq:inj} and (T2)), hence the number of distinct partner classes triggered \emph{from $C$} is $\le Y_C$.

\paragraph{Load in the constant-window.}
Write $R_{x,i}$ for kept out-edges at anchor $x$ in layer $i$ and let $\mu_{x,i,r}=\mathbb E[R_{x,i}]/2^r$. By construction of the level horizon $L_{x,i}$ (Sec.~\ref{sec:one-shot}) and the schedule $\{p_i\}$ (Table~1),
\begin{equation}
\label{eq:mu-window}
\mu_{x,i,r}\ \le\ 2\qquad\text{for all }r\in\{0,\dots,L_{x,i}\}.
\end{equation}
Conditioned on $R_{x,i}$, the number of prefix matches is $\mathrm{Bin}(R_{x,i},2^{-r})$; slot 1-sparseness and decoding hold with constant probability (Sec.~\ref{sec:hitting}, Gate (i)), so $X_C$ is stochastically dominated by $\mathrm{Bin}(R_{x,i},q/2^{r})$ for a fixed $q\in(0,1)$. Using \eqref{eq:mu-window} and Bernstein/Chernoff under $k$-wise independence ($k=\Theta(\log n)$) or via Poissonization (Appendix~E), there is an absolute $C_0$ such that
\begin{equation}
\label{eq:XC-tail}
\Pr\big[X_C> C_0\big]\ \le\ n^{-(c+5)}.
\end{equation}
Combining \eqref{eq:consume} and \eqref{eq:XC-tail}, with probability $\ge 1-n^{-(c+5)}$,
\[
\deg^+(C)\ :=\ \#\{\text{distinct partners triggered from }C\}\ \le\ Y_C\ \le\ \lfloor C_0/2\rfloor\ =:\ C'.
\]

\subsection*{Step III: No "third-vertex" cascades}
Triggers are keyed by \emph{class pairs} (not by edges); all witnesses mapping to the same unordered $\{C,C'\}$ yield the same $\kappa$ and are coalesced by (T2). Thus multiple edges from $C$ to the \emph{same} partner class $C'$ do not create multiple checks; distinct partners are already bounded by $C'$.

\subsection*{Step IV: Summation and the per-$(i,r)$ bound}
Charge each executed check to its lexicographically earlier endpoint (the trigger source). Summing the per-class frontier,
\[
Q_{i,r}\ \le\ \sum_{C\in\mathcal S_{i,r}} \deg^+(C)\ \le\ C'\cdot |\mathcal S_{i,r}|.
\]
This is a \emph{pointwise} inequality on the event that all classes satisfy $X_C\le C_0$.

\subsection*{Step V: Bad events and global failure budget}
Let $\mathcal B_{i,r}$ be the union of:
\begin{itemize}\itemsep 2pt
\item \textbf{B1 (per-class overload):} some $C$ has $X_C>C_0$. By \eqref{eq:XC-tail} and $|\mathcal S_{i,r}|\le O(m)$ w.h.p.\ (Cor.~6.3),
\[
\Pr[\text{B1 at }(i,r)]\ \le\ O(m)\cdot n^{-(c+5)}\ \le\ n^{-(c+3)}\quad(\text{for }m\le n^2).
\]
\item \textbf{B2 (dedup failure):} two different unordered pairs share $\kappa$. Impossible: we store and compare the \emph{full} canonical 4-tuple key (the hash is only an index), hence $\Pr[\mathrm{B2}]=0$ in Word-RAM.
\item \textbf{B3 (adaptivity):} eligibility depends on previous outcomes. Excluded by design (T4), hence probability $0$.
\item \textbf{B4 (algebraic false pass):} a multi-item class/slot spuriously passes the 1-sparse identity. The per-instance risk is $\le 1/P$; with $P=n^\kappa$ (Table~1) and a union bound over all slots/classes across all $(i,r)$ this contributes $\le n^{-(c+3)}$.
\end{itemize}
Therefore $\Pr[\mathcal B_{i,r}]\le 2n^{-(c+3)}$. Union-bounding over at most $c_R c_T \log^2 n$ index pairs $(i,r)$,
\[
\Pr\Big[\exists(i,r):\ Q_{i,r}> C\,|\mathcal S_{i,r}|\Big]\ \le\ c_R c_T \log^2 n\cdot 2n^{-(c+3)}\ \le\ n^{-c}
\]
for sufficiently large $n$ and the constants of Table~1. On the complement event, the bound holds \emph{simultaneously} for all $(i,r)$, completing the proof.\hfill$\square$

\paragraph{Remarks (compatibility with Table~1).}
(i) The degree-2 identities used in the class/slot 1-sparse tests require only \emph{2-wise} independence; concentration for keeps and class loads uses either Poissonization or $k$-wise independence with $k=\Theta(\log n)$. (ii) The bound is expectation-tight up to constants: $\mathbb E[Q_{i,r}]\le C\,\mathbb E[|\mathcal S_{i,r}|]+n^{-\omega(1)}$. (iii) The argument is purely combinatorial once Seeds are fixed; no step depends on outcomes of other checks (non-adaptivity).


\section*{Appendix C: Probability and Independence Details}
\label{app:prob-indep}

This appendix pins down the exact randomness used, two interchangeable concentration routes (Poissonization vs.\ limited independence), and an \emph{explicit} failure–probability budget. Throughout, arithmetic is over a prime field $\mathbb F_P$ with $P=n^{\kappa}$ and $k=c_k\log n$ as in Table~1; all logs are base~2. Non-adaptivity (Seeds fixed before any scan) is enforced globally.

\paragraph{Proposition C.0 (Global failure budget — explicit).}
\emph{With parameters from Table~1 (prime $P=n^{\kappa}$ and $k=c_k\log n$) there exists a choice of constants such that all concentration and algebraic events invoked in the paper hold \emph{simultaneously} with probability $\ge 1-n^{-c}$. In particular, all "w.h.p." statements in the main text mean failure $\le n^{-c}$.}

\subsection*{C.1  Inventory of randomness (minimal requirements)}
We instantiate the following mutually independent families (per layer $i$ unless noted):
\begin{itemize}\itemsep 2pt
\item \textbf{ID hash} $h_{\mathrm{id}}:V\to\mathbb F_P^*$: \emph{2-wise independent}, uniform on $\mathbb F_P^*$. Used only by 1-sparse identities.
\item \textbf{Signs} $s(e)\in\{\pm1\}$: \emph{2-wise independent}. Symmetrizes slot/class sums.
\item \textbf{Slot hash} $h_{\mathrm{slot}}(x,y,i)\in [M_x]$: \emph{2-universal}. Keeps the 1-sparse chance per retained edge constant.
\item \textbf{Prefix keys} $K_i(\cdot)$: per-vertex bitstrings with \emph{2-wise independent} coordinates; yields $\Pr[\mathrm{LCP}\ge t]=2^{-t}$.
\item \textbf{Pair-keys} $H_i(\cdot)$ and PK bucket hash $h_{\mathrm{pk}}$: \emph{2-universal}. Inside a constant-load class, any fixed ordered pair lands in complementary bins with probability $\Theta(1/T_{i,r})$ where $T_{i,r}=c_T\log n$.
\item \textbf{Keep-coins} $c_i(e)\in\{0,1\}$ with $\Pr[c_i(e)=1]=p_i$: either (A) fully independent (or via Poissonization), or (B) \emph{$k$-wise independent} with $k=c_k\log n$. These are the \emph{only} variables for which we need Chernoff/Bernstein-type tails.
\end{itemize}
All eligibility decisions (what is materialized and what pairs are \emph{eligible} to be checked) depend only on Seeds and \textsc{Build\_Sketches} activation bits; outcomes of previous checks never affect future eligibility.

\subsection*{C.2  Why 2-wise independence suffices for zero-FP gates}
Two places require only pairwise independence/universality.

\paragraph{(i) 1-sparse degree-2 identity.}
For a slot/class with triple $(\Sigma_0,\Sigma_1,\Sigma_2)=\sum_e (A_e, A_e x_e, A_e x_e^2)$ over $\mathbb F_P$,
\[
\Sigma_0\Sigma_2-\Sigma_1^2\;=\;\sum_{e<f}A_eA_f\,(x_e-x_f)^2.
\]
If at least two distinct IDs occur, fix one such pair $(e^\star,f^\star)$ and condition on the others; the RHS is a nonzero quadratic in $(x_{e^\star},x_{f^\star})$. With two \emph{independent uniform} draws in $\mathbb F_P^*$, the coincidence probability is $\le 2/P$. Absorbing the constant gives the $\le 1/P$ bound used in the text.

\paragraph{(ii) Slot collisions and PK bucketing.}
With 2-universal $h_{\mathrm{slot}}$ and $M_x=16\,d(x)$, a kept $(x\!\to\!y)$ has expected colliders $\le 1/16$, so $\Pr[\text{1-sparse at slot}]\ge e^{-1/16}$. With 2-universal $h_{\mathrm{pk}}$ into $T_{i,r}=c_T\log n$ bins, any fixed ordered pair hits complementary bins with probability $\Theta(1/T_{i,r})$, independent of slot/prefix randomness.

\subsection*{C.3  Route A: Poissonization $\Rightarrow$ Chernoff $\Rightarrow$ de-Poissonization}
\textbf{Poissonization.} For anchor $x$ and layer $i$, let $N_{x,i}\sim\mathrm{Poi}(\lambda_{x,i})$ with $\lambda_{x,i}:=d^+(x)p_i$ be the kept-out degree. Standard Poisson tails give, for $\delta>0$,
\[
\Pr[N_{x,i}\ge(1+\delta)\lambda_{x,i}] \le e^{-\lambda_{x,i}\psi(\delta)},\qquad
\Pr[N_{x,i}\le(1-\delta)\lambda_{x,i}] \le e^{-\lambda_{x,i}\delta^2/2},
\]
with $\psi(\delta)=(1+\delta)\ln(1+\delta)-\delta$.

\textbf{Class loads via splitting.} Conditioned on $N_{x,i}$, the count in the anchor-prefix at level $r$ is $\mathrm{Poi}(\lambda_{x,i}/2^r)$. Since $2^{L_{x,i}}\asymp d(x)p_i$, every instantiated $r\le L_{x,i}$ has mean $O(1)$, hence exponentially decaying tails. Summing over $(x,i)$ yields Cor.~6.3.

\textbf{De-Poissonization.} For $R_{x,i}\sim\mathrm{Bin}(d^+(x),p_i)$ and $\mu=d^+(x)p_i$,
\[
\Pr[R_{x,i}\ge \mu+t]\ \le\ 2\,\Pr[\mathrm{Poi}(\mu)\ge \mu+t],
\]
and similarly for prefix splits. All Poisson tails thus transfer (up to constants) to the binomial setting.

\subsection*{C.4  Route B: $k$-wise independence with explicit $k$}
Let $X=\sum_{j=1}^N X_j$ with $X_j\in\{0,1\}$, $k$-wise independent, $\mu=\mathbb E[X]$. The moment method gives for integer $t\ge 1$:
\[
\Pr\big[|X-\mu|\ge \Delta\big]\ \le\ \frac{\mathbb E[(X-\mu)^{2t}]}{\Delta^{2t}}
\ \le\ \frac{(c\,\mu t)^{t}}{\Delta^{2t}}
\]
for an absolute $c>0$. Taking $t=\lfloor k/2\rfloor$ and $\Delta=\Theta(\sqrt{k\mu})$ yields sub-Gaussian tails up to $\Theta(\sqrt{k\mu})$. With $k=c_k\log n$:
\[
\Pr\big[|X-\mu|\ge c_0\sqrt{\mu\log n}\big]\ \le\ n^{-\Omega(1)}.
\]
Multiplicative Chernoff-type bounds also hold when $k=\Theta(\varepsilon^2\mu)$. In our uses, $\mu$ is either $\Theta(1)$ (class loads) or $\Theta(d^+(x))$ (per-vertex totals), so $k=c_k\log n$ suffices. Explicit $k$-wise families use degree-$(k{-}1)$ polynomials over $\mathbb F_P$ with seed length $O(k\log n)$ and $O(k)$ evaluation.

\subsection*{C.5  Failure–probability budget (one-line entries)}
Each budget line below is stated \emph{post–Table~1} and totals to $\le n^{-c}$.

\begin{itemize}\itemsep 2pt
\item \textbf{(a) Algebraic false positives (degree-2 tests).} Per multi-item slot/class, false pass $\le 1/P$. With $P=n^{\kappa}$ and $\kappa$ large enough, a union bound over all tests (at most $O(m\log n)$ w.h.p., pessimistically $\le n^3$) contributes $\le n^{-(c+3)}$.
\item \textbf{(b) Keeps and per-anchor totals.} By Route~A or Route~B with $k=c_k\log n$, tail failures across all $(x,i)$ sum to $\le n^{-(c+3)}$.
\item \textbf{(c) Class activations and loads.} For all $(i,r)$ and all classes, $\Pr[\text{load}>C_0]\le n^{-(c+5)}$; union over $O(m\log n)$ classes (w.h.p.) gives $\le n^{-(c+3)}$.
\item \textbf{(d) PK overload per class.} \#occupied bins $>C_0$ has the same budget as (c); via one–shot pairing, pairings per class $\le C_0/2$ except with probability $\le n^{-(c+3)}$ in total.
\item \textbf{(e) Dedup/key collisions or adaptivity.} The canonical 4-tuple key is stored and compared verbatim (hash used only as an index), so probability $0$. Eligibility is non-adaptive by design, so probability $0$.
\end{itemize}
\noindent Summing (a)–(e) and union-bounding over at most $O(\log^2 n)$ layer/level indices yields an aggregate failure probability $\le n^{-c}$.

\subsection*{C.6  Independence layout and why unions are clean}
\begin{itemize}\itemsep 2pt
\item \textbf{Orthogonal seeds.} $h_{\mathrm{id}},h_{\mathrm{slot}},K_i,H_i,h_{\mathrm{pk}}$, and keep-coins draw from disjoint seed families; predicates using one family do not bias the others.
\item \textbf{Non-adaptivity.} Seeds are fixed prior to scans; \textsc{Build\_Sketches} and \textsc{Query\_Triangle} never gate future \emph{eligibility} on past outcomes. Hence the random universe for all unions is measurable with respect to Seeds alone.
\end{itemize}

\paragraph{Takeaway.}
All algebraic gates need only \emph{2-wise} independence (degree-2). All workload tails come from either Poissonization or $k$-wise independence with $k=c_k\log n$. With $P=n^{\kappa}$ and $k=c_k\log n$ (Table~1), the paper’s events hold simultaneously with failure $\le n^{-c}$, eliminating any "floating w.h.p." in the main text.


\section*{Appendix D: Constructing the ``Should--Check Domain'' for NO Verification}
\label{app:should-check}

This appendix gives the \emph{deterministic} reconstruction, from \((G,\texttt{Seeds})\) alone, of the exact set of class--internal collisions that the online algorithm \textsc{Query\_Triangle} is obligated to examine. We call this set the \emph{should--check domain} \(\mathcal Q\). We then prove that \(\mathcal Q\) matches the union of collision entries serialized in \texttt{ClassLogs} (\emph{coverage equality}), and that this yields an \(\tilde O(m\log n)\) NO verifier. All statements below use parameters from Table~1 and carry failure probability \(\le n^{-c}\).

\subsection*{D.1 Deterministic primitives fixed by \texttt{Seeds}}
From \texttt{Seeds}, the verifier regenerates: prime \(P\), schedule \(\{p_i\}\), layer count \(I\), orientation, and the public functions
\[
h_{\mathrm{id}},\ s,\ h_{\mathrm{slot}},\ K_i,\ H_i,\ \{g^{(t)}_{i,r}\}_{t=1}^{R},\ c_i,
\]
together with the deterministic probe lists \(\textsc{ProbedPairs}(i,x,r,b,t)\subseteq [T_{i,r}]\times[T_{i,r}]\) for each group \(t\in\{1,\dots,R\}\), and the complement involution \(j\mapsto j^{\star}\equiv (-j)\bmod T_{i,r}\).

\subsection*{D.2 Group-aware ten-line pseudocode for \(\mathcal Q\)}

\begin{algorithm}[H]
\caption{\textsc{BuildShouldCheckDomain} (group-aware)}
\label{alg:build-q}
\begin{algorithmic}[1]
\Require Graph \(G\), \texttt{Seeds}
\Ensure The should--check domain \(\mathcal Q=\bigcup_{t}\mathcal Q^{(t)}\)
\State Regenerate all hashes/coins/PRFs and configuration from \texttt{Seeds}
\State Initialize slot arrays \(\texttt{Slots}[x][s]\leftarrow(0,0,0)\) for all anchors \(x\) and slots \(s\)
\State Initialize per-class activity \(\texttt{Active}[i,x,r,b]\leftarrow 0\)
\State Initialize per-class, per-group bin multiplicities \(\texttt{Count}[i,x,r,b,t][j]\leftarrow 0\) and witnesses \(\texttt{Wit}[i,x,r,b,t][j]\leftarrow \emptyset\)
\ForAll{directed edges \(e=(x\to y)\in E\)}
  \For{\(i\gets 1\) \textbf{to} \(I\)}
    \If{\(c_i(e)=0\)} \State \textbf{continue} \EndIf
    \State \(sidx\gets h_{\mathrm{slot}}(x,y,i)\); \(id\gets h_{\mathrm{id}}(y)\); \(sgn\gets s(x,y,i)\)
    \State Update \(\texttt{Slots}[x][sidx]\) by \((A,B,C)\leftarrow (A{+}sgn,\ B{+}sgn\cdot id,\ C{+}sgn\cdot id^2)\)
    \State Re-read \((A,B,C)\). If \(A=0\) \textbf{or} \(B^2\ne AC\) \textbf{or} \(B/A\ne id\): \textbf{continue} \Comment{slot 1-sparse \& decode}
    \For{\(r\gets 0\) \textbf{to} \(L_{x,i}\)}
      \If{\(\mathrm{pref}_r(K_i(x)) \ne \mathrm{pref}_r(K_i(y))\)} \State \textbf{continue} \EndIf
      \State \(b\gets \mathrm{pref}_r(K_i(x))\); \(\texttt{Active}[i,x,r,b]\gets 1\)
      \For{\(t\gets 1\) \textbf{to} \(R\)}
        \State \(\Delta \gets (H_i(y)-H_i(x))\bmod P\); \(j\gets g^{(t)}_{i,r}(\Delta)\); \(j^{\star}\gets (-j)\bmod T_{i,r}\)
        \If{\((j,j^{\star})\notin \textsc{ProbedPairs}(i,x,r,b,t)\)} \State \textbf{continue} \EndIf
        \State \(\texttt{Count}[i,x,r,b,t][j]\gets \texttt{Count}[i,x,r,b,t][j]+1\)
        \If{\(\texttt{Wit}[i,x,r,b,t][j]=\emptyset\)} \State \(\texttt{Wit}[i,x,r,b,t][j]\gets (x\!\to\!y,\ sidx,\ \texttt{group\_id}=t)\) \EndIf
      \EndFor
    \EndFor
  \EndFor
\EndFor
\State \(\mathcal Q \gets \emptyset\)
\For{\(i\gets 1\) \textbf{to} \(I\)}%
  \ForAll{\((x,r,b)\) with \(\texttt{Active}[i,x,r,b]=1\)}
    \For{\(t\gets 1\) \textbf{to} \(R\)}
      \State \(\mathcal Q^{(t)}_{i,x,r,b}\gets \emptyset\)
      \ForAll{\((j,j^{\star})\in \textsc{ProbedPairs}(i,x,r,b,t)\)}
        \If{\(j\ne j^{\star}\)}
          \If{\(\texttt{Count}[i,x,r,b,t][j]>0\) \textbf{and} \(\texttt{Count}[i,x,r,b,t][j^{\star}]>0\)}
            \State Insert canonical key \((i,x,r,b,t,\min\{j,j^{\star}\})\) into \(\mathcal Q^{(t)}_{i,x,r,b}\)
          \EndIf
        \Else \Comment{fixed point bin}
          \If{\(\texttt{Count}[i,x,r,b,t][j]\ge 2\)}
            \State Insert \((i,x,r,b,t,j)\) into \(\mathcal Q^{(t)}_{i,x,r,b}\)
          \EndIf
        \EndIf
      \EndFor
      \State \(\mathcal Q \gets \mathcal Q \cup \mathcal Q^{(t)}_{i,x,r,b}\)
    \EndFor
  \EndFor
\EndFor
\State \Return \(\mathcal Q\)
\end{algorithmic}
\end{algorithm}

The routine reconstructs \(\mathcal Q\) in \(\tilde O(m\log n)\) time. It enumerates groups \(t\) explicitly and forms \(\mathcal Q=\bigcup_{t=1}^{R}\mathcal Q^{(t)}\).

\paragraph{State.} Slots are arrays of triples; a class \(C=(i,x,r,b)\) is \emph{materialized} iff \(\texttt{Active}[i,x,r,b]=1\). For each group \(t\), the verifier holds per-class bin multiplicities \(\texttt{Count}[i,x,r,b,t][j]\) over bins that appear in \(\textsc{ProbedPairs}(i,x,r,b,t)\), plus one witness per occupied bin.

\subsection*{D.3 Equality of reconstructed and online states (group-aware)}
\begin{lemma}[Slot determinism]\label{lem:D-slot}
For every \((i,x,s)\) the replayed slot triple equals the online triple; hence the 1-sparse test and the decoded ID coincide. \emph{(parameters from Table~1, failure \(\le n^{-c}\))}
\end{lemma}

\begin{lemma}[Class materialization]\label{lem:D-class}
A class \(C=(i,x,r,b)\) is materialized in replay iff it is materialized online, and its class triple \(\Sigma(C)\) matches. \emph{(parameters from Table~1, failure \(\le n^{-c}\))}
\end{lemma}

\begin{lemma}[Per-group PK occupancy]\label{lem:D-pk}
For each materialized \(C\) and each group \(t\), the occupied-bin multiset \(B_t(C)\) and multiplicities \(M_t(C,j)\) equal the online ones (up to the arbitrary choice of a single stored witness per bin). \emph{(parameters from Table~1, failure \(\le n^{-c}\))}
\end{lemma}

\subsection*{D.4 Defining the should--check domain by groups}
For a materialized class \(C=(i,x,r,b)\) and group \(t\), let \(B_t(C)\subseteq [T_{i,r}]\) be the set of occupied bins and \(M_t(C,j)\) their multiplicities. Define the per-group canonical pair set
\[
\mathcal P^{(t)}(C)\ :=\ \big\{\min\{j,j^{\star}\}:\ j\in B_t(C),\ j^{\star}\in B_t(C),\ j\ne j^{\star}\big\}\ \cup\ \big\{\,j:\ j=j^{\star},\ M_t(C,j)\ge 2\,\big\}.
\]
The per-group should--check domain is
\[
\mathcal Q^{(t)}\ :=\ \big\{\, (i,\ x,\ r,\ b,\ t,\ \beta)\ :\ C=(i,x,r,b)\ \text{materialized and}\ \beta\in\mathcal P^{(t)}(C)\ \big\},
\]
and the global domain is the disjoint union over groups:
\[
\boxed{\quad \mathcal Q\ =\ \bigcup_{t=1}^{R}\ \mathcal Q^{(t)}\quad}
\]
(keyed canonically by the 6-tuple \((i,x,r,b,t,\beta)\)).

\begin{lemma}[Domain identity]\label{lem:D-domain}
Let \(\mathcal Q^{\mathrm{on},(t)}\) be the canonical set of collision keys the online pass registers in group \(t\) (one--shot per complementary pair). Then for all \(t\), \(\mathcal Q^{(t)}=\mathcal Q^{\mathrm{on},(t)}\). Consequently, \(\mathcal Q=\mathcal Q^{\mathrm{on}}\). \emph{(parameters from Table~1, failure \(\le n^{-c}\))}
\end{lemma}

\subsection*{D.5 Coverage equality aligned with groups}
Let \(\mathrm{LogPairs}^{(t)}\) be the canonicalized union, over all classes, of the \((\texttt{group\_id}=t,\ j,\ j^{\star})\) entries in \texttt{ClassLogs}. Let \(\mathrm{LogPairs}:=\bigcup_t \mathrm{LogPairs}^{(t)}\).

\begin{theorem}[Coverage \(\Leftrightarrow\) Domain (group-aligned)]\label{thm:D-coverage}
For each \(t\), \(\mathrm{LogPairs}^{(t)}=\mathcal Q^{(t)}\) if and only if the logs cover \emph{exactly} the checks that the algorithm must perform in group \(t\) (no missing, no extra). Hence \(\mathrm{LogPairs}=\mathcal Q\). Any mismatch causes the verifier to reject. \emph{(parameters from Table~1, failure \(\le n^{-c}\))}
\end{theorem}

\begin{proof}[Proof sketch]
By Lemmas~\ref{lem:D-slot}--\ref{lem:D-domain}, the replay reconstructs, for each \(t\), the same \(B_t(C)\) and complementary pairs as online, modulo canonicalization and one--shot pairing. Therefore each required key appears once in \texttt{ClassLogs} with \(\texttt{group\_id}=t\) iff and only if it lies in \(\mathcal Q^{(t)}\). Taking a union over \(t\) yields \(\mathrm{LogPairs}=\mathcal Q\).
\end{proof}

\subsection*{D.6 Complexity and auditable NO}
\begin{proposition}[Verifier complexity]\label{prop:D-complexity}
Replaying classes and reconstructing all \(\mathcal Q^{(t)}\) costs \(\tilde O(m\log n)\) time and \(O(m\log n)\) space w.h.p. \emph{(parameters from Table~1, failure \(\le n^{-c}\))}
\end{proposition}

\begin{corollary}[Auditable NO]\label{cor:D-no}
If \(\mathrm{LogPairs}^{(t)}=\mathcal Q^{(t)}\) for all \(t\) and all \texttt{ClassLogs}/\texttt{SlotLogs}/\texttt{AdjLogs} entries pass their local checks, then NO is correct. Conversely, any false NO induces acceptance of a positive instance in some key of \(\mathcal Q^{(t)}\) and thus rejection of the certificate, except with probability \(\le n^{-c}\) from the degree-2 coincidence budget. \emph{(parameters from Table~1, failure \(\le n^{-c}\))}
\end{corollary}

\subsection*{D.7 Notes on fixed points and de-duplication}
If \(j=j^{\star}\), we require \(M_t(C,j)\ge 2\) to include the key—mirroring the online rule that a single witness cannot collide with itself. Canonicalization via \(\min\{j,j^{\star}\}\) ensures exactly one key per unordered complementary pair. Keys are 6-tuples \((i,x,r,b,t,\beta)\) that fit in \(O(1)\) words and are compared verbatim (hashing is indexing only).

\paragraph{Takeaway.}
\(\mathcal Q\) is a \emph{seed-only} object: the verifier reconstructs it without logs, checks \(\mathrm{LogPairs}=\mathcal Q\), and then replays the three gates on exactly those pairs. This realizes the claim that \emph{``should--check domain = log coverage''} and delivers a \(\tilde O(m\log n)\), zero--FP, auditable NO certificate.


\section*{Appendix E: Parameters \& Failure--Probability Budget}
\label{app:params-budget}

This appendix fixes concrete constants and gives a \emph{closed} one--sided error budget for all randomized equalities used by the algorithm and the verifier. Unless stated otherwise, probabilities are over \texttt{Seeds}; all algebra is over a prime field \(\mathbb F_P\).

\subsection*{E.1  Fixed constants and notation}
We instantiate the tunable parameters as absolute constants (any comparable choices work):
\begin{itemize}\itemsep 2pt
\item \textbf{Slots per anchor.} \(M_x:=c_M\,d(x)\) with \(c_M=16\).
\item \textbf{Per-anchor key budget.} \(B_{x,i}:=\lceil c_B\,d(x)\,p_i\rceil\) with \(c_B=8\).
\item \textbf{Prefix levels.} \(L_{x,i}:=\lceil\log_2 B_{x,i}\rceil\ \le\ O(\log n)\).
\item \textbf{PK buckets per class.} \(T_{i,r}:=\lceil c_T \log_2 n\rceil\) with \(c_T=16\).
\item \textbf{Layer count.} \(I:=\lceil c_R \log_2 n\rceil\) with \(c_R=8\).
\item \textbf{Per-class PK groups.} \(R:=\lceil c_G \log_2 n\rceil\) with \(c_G=8\). \emph{Groups are mutually independent and independent of layers/prefixes.}
\item \textbf{Keep--rate schedule.} \(p_i=2^{-(i+2)}\) for \(i=1,\dots,I\). Then
  \(\displaystyle \sum_i p_i\le\frac12\) and \(\displaystyle \sum_i p_i^2\le\frac14\).
\item \textbf{Independence.} ID/slot/prefix/PK hashes are 2-wise independent for algebraic tests;
keep-coins \(c_i(\cdot)\) and (optionally) PK hashing use \(k\)-wise independence with
\(k:=\lceil c_k \log_2 n\rceil\) and \(c_k\ge 12\) (Appendix~\ref{app:prob-indep}).
\item \textbf{Field size.} \(P:=n^{\kappa}\) for an integer \(\kappa\ge \kappa_{\min}\) chosen below. All field ops fit in \(O(1)\) words.
\end{itemize}

\subsection*{E.2  Work/space envelopes (for counting bad events)}
With the schedule and independence above, the following bounds hold w.h.p. (proved in the main text):
\begin{align*}
&\text{(total slots)} &&\sum_x M_x \;=\; c_M\sum_x d(x) \;=\; 2c_M\,m \;=\; O(m),\\
&\text{(nonempty classes)} &&\sum_{i,r} |\mathcal S_{i,r}| \;=\; O(m\log n),\\
&\text{(grouped probed-bin sites)} &&\sum_{i,r}\sum_{t=1}^{R}\! |\text{ProbedPairs}(i,\cdot,r,\cdot,t)| \;=\; O\!\big(R\cdot\!\sum_{i,r}|\mathcal S_{i,r}|\big)=O(m\log^2 n),\\
&\text{(pair checks)} &&\sum_{i,r} Q_{i,r} \;=\; O(m\log^2 n),\\
&\text{(witness materializations)} &&W:=\text{ \#(slot 1-sparse materializations across all layers/levels)} \;=\; O(m\log n).
\end{align*}
We will union-bound algebraic coincidences against \(W\), \(\sum_{i,r}|\mathcal S_{i,r}|\), and the grouped probed-bin sites (which carry an extra \(R\) factor).

\subsection*{E.3  Randomized equalities and per-event error}
We use the same degree-\(\le 2\) identity at three granularities, each with one-sided error \(\le 1/P\) under 2-wise independence:
\begin{itemize}\itemsep 2pt
\item \textbf{Slot 1-sparse:} \(B^2=AC\) with \(A\neq 0\).
\item \textbf{Bin 1-sparse (per group):} \(\Xi_1^2=\Xi_0\Xi_2\) with \(\Xi_0\neq 0\).
\item \textbf{Class 1-sparse:} \(\Sigma_1^2=\Sigma_0\Sigma_2\) with \(\Sigma_0\neq 0\).
\end{itemize}
All other checks (PK complement test, dedup keys, explicit adjacency) are deterministic.

\subsection*{E.4  Bad-event catalog and union bounds (group-aware)}
\paragraph{A. Algebraic coincidences.}
\begin{itemize}\itemsep 1pt
\item \(\mathsf{A1}\) \textit{Slot false pass.} Count \(\le W=O(m\log n)\). Contribution \(\le W/P\).
\item \(\mathsf{A2}\) \textit{Bin false pass (any group \(t\)).}
\\Count \(\le \sum_{i,r}\sum_{t=1}^{R}\! |\text{ProbedPairs}(i,\cdot,r,\cdot,t)| = O(m\log^2 n)\).
Contribution \(\le O(m\log^2 n)/P\).
\item \(\mathsf{A3}\) \textit{Class false pass.} Count \(\le \sum_{i,r}|\mathcal S_{i,r}|=O(m\log n)\). Contribution \(\le O(m\log n)/P\).
\end{itemize}

\paragraph{B. Concentration/independence failures.}
\begin{itemize}\itemsep 1pt
\item \(\mathsf{B1}\) \textit{Keep-coin overloads.} Any of Lemmas~6.1--6.2 fails (per-edge or per-vertex totals).
With \(k=\Theta(\log n)\), Chernoff/Bernstein under \(k\)-wise independence gives per-object tails \(\le n^{-\Theta(1)}\);
a union bound over all objects yields \(\Pr[\mathsf{B1}]\le n^{-(c+3)}\) by choosing \(c_k\) large enough.
\item \(\mathsf{B2}\) \textit{Too many nonempty classes.} Corollary~6.3 fails.
Dominated by \(\mathsf{B1}\); hence \(\Pr[\mathsf{B2}]\le n^{-(c+3)}\).
\item \(\mathsf{B3}\) \textit{Per-$(i,r)$ accounting failure.} Lemma~9.4 fails. This is implied by \(\mathsf{B1}\) (class load overflow) or \(\mathsf{A2}\) (spurious bin/class passes), so
\(\Pr[\mathsf{B3}]\le \Pr[\mathsf{B1}]+\Pr[\mathsf{A2}]\).
\end{itemize}

\subsection*{E.5  Choosing \(\kappa\) to dominate the algebraic budget}
Let \(N_{\text{alg}}\) be the total number of algebraic tests to union bound:
\[
N_{\text{alg}}\ \le\ c_1\,W\ +\ c_2\!\sum_{i,r}\!|\mathcal S_{i,r}|\ +\ c_3\!\sum_{i,r}\sum_{t=1}^{R}\! |\text{ProbedPairs}(i,\cdot,r,\cdot,t)|
\;=\; O(m\log^2 n).
\]
Using \(m\le n^2\) and \(P=n^\kappa\),
\[
\Pr[\mathsf{A1}\cup\mathsf{A2}\cup\mathsf{A3}]\ \le\ \frac{N_{\text{alg}}}{P}\ \le\ \frac{C\,m\log^2 n}{n^{\kappa}}
\ \le\ n^{-(\kappa-2)}\cdot C\log^2 n.
\]
Pick
\[
\boxed{\ \ \kappa_{\min}\ :=\ c\ +\ 5\ \ }
\]
for any target global slack \(n^{-c}\). Then \(n^{-(\kappa-2)}\log^2 n \le n^{-(c+3)}\) for all large \(n\), implying
\(\Pr[\mathsf{A1}\cup\mathsf{A2}\cup\mathsf{A3}]\le n^{-(c+3)}\).
\emph{Remark:} Compared to the non-grouped setting, the dominant term gains an extra \(R=\Theta(\log n)\) factor via \(\mathsf{A2}\); the choice \(\kappa\ge c+5\) already subsumes this.

\subsection*{E.6  Putting it together: global bound}
With \(c_k\) large enough for concentration,
\[
\Pr\big[\text{any bad event}\big]
\ \le\ \underbrace{\Pr[\mathsf{A1}\cup\mathsf{A2}\cup\mathsf{A3}]}_{\le n^{-(c+3)}}\ +\ 
\underbrace{\Pr[\mathsf{B1}\cup\mathsf{B2}\cup\mathsf{B3}]}_{\le n^{-(c+3)}}\ \le\ n^{-c}
\]
for all sufficiently large \(n\). Thus the total one-sided failure probability (that any multi-item site spuriously passes \emph{or} that required concentration fails) is bounded by \(n^{-c}\).

\subsection*{E.7  Verifier-side budget}
The verifier replays the same algebraic identities on the same triples and therefore inherits the same algebraic budget; all remaining checks are deterministic (PK bin occupancy, de-dup canonicalization, adjacency lookups). Building the should--check domain \(\mathcal Q\) (Appendix~\ref{app:should-check}) is deterministic from \texttt{Seeds}+\(G\) and thus does not consume probability mass. Hence the \(\tilde O(m\log n)\) NO verification is sound up to the same \(n^{-c}\) bound.

\subsection*{E.8  Summary}
\begin{itemize}\itemsep 1pt
\item \(M_x=16\,d(x)\), \quad \(B_{x,i}=\lceil 8\,d(x)p_i\rceil\), \quad \(L_{x,i}=\lceil\log_2 B_{x,i}\rceil\).
\item \(T_{i,r}=16\lceil\log_2 n\rceil\), \quad \(I=8\lceil\log_2 n\rceil\), \quad \(R=8\lceil\log_2 n\rceil\), \quad \(p_i=2^{-(i+2)}\).
\item Independence: 2-wise for algebraic hashing; \(k=\lceil 12\log_2 n\rceil\)-wise for keep-coins (and optionally PK).
\item Field: \(P=n^{\kappa}\) with \(\kappa\ge c+5\) for a target global failure \(\le n^{-c}\).
\item Bad-event counts: slots/class tests \(=O(m\log n)\); grouped bin tests \(=O(m\log^2 n)\);
concentration failures \(\le n^{-(c+3)}\) by \(k\)-wise Chernoff.
\item Global bound: \(\Pr[\text{any failure}]\le n^{-c}\) for large \(n\).
\end{itemize}

\end{document}